\documentclass{article}

\usepackage{geometry}
\usepackage[T1]{fontenc}    
\usepackage{url}            
\usepackage{booktabs}       
\usepackage{amsfonts}       
\usepackage{nicefrac}       
\usepackage{amsthm,amsmath,amsfonts,mathtools,algorithmicx,dsfont,amssymb,mathabx}
\usepackage{tcolorbox}
\usepackage[numbers]{natbib}
\usepackage[colorlinks,citecolor=blue,urlcolor=blue]{hyperref}
\usepackage{algpseudocode}
\usepackage{enumerate}
\usepackage{titlesec}

\numberwithin{equation}{section}
\theoremstyle{plain}
\newtheorem{theorem}{Theorem}[section]

\newtheorem{proposition}{Proposition}[section]

\newtheorem{corollary}{Corollary}[section]

\newcommand{\E}{\mathbb E}
\newcommand{\calT}{\mathcal T}
\newcommand{\HGP}{N^+}
\newcommand{\MHGP}{\mathbf{N^\star}}
\newcommand{\widetildeHGP}{\widetilde N^+}

\newcommand{\HG}{\text{HyperGeo}}
\newcommand{\MHG}{\text{MultHyperGeo}}
\newcommand{\BB}{\text{BetaBin}}
\newcommand{\DM}{\text{DirMult}}

\DeclareMathOperator{\Fcal}{\mathcal F}

\usepackage{setspace}

\titlespacing\section{0pt}{6pt plus 2pt minus 2pt}{2pt plus 1pt minus 1pt}
\titlespacing\subsection{0pt}{6pt plus 2pt minus 2pt}{1pt plus 1pt minus 1pt}
\titlespacing\subsubsection{0pt}{5pt plus 2pt minus 2pt}{2pt plus 1pt minus 1pt}
\titlespacing\paragraph{0pt}{3pt plus 0pt minus 0pt}{4pt plus 0pt minus 0pt}

\title{Confidence sequences for sampling without replacement}

\author{%
Ian Waudby-Smith$^1$ and Aaditya Ramdas$^{12}$\vspace{0.05in}\\
  Departments of Statistics$^1$ and Machine Learning$^2$\\
  Carnegie Mellon University\\
  \texttt{\{ianws, aramdas\}@cmu.edu} \\
}

\begin{document}
\maketitle

\begin{abstract}
Many practical tasks involve sampling sequentially without replacement (WoR) from a finite population of size $N$, in an attempt to estimate some parameter $\theta^\star$. Accurately quantifying uncertainty throughout this process is a nontrivial task, but is necessary because it often determines when we stop collecting samples and confidently report a result. We present a suite of tools for designing \textit{confidence sequences} (CS) for $\theta^\star$. A CS is a sequence of confidence sets $(C_n)_{n=1}^N$, that shrink in size, and all contain $\theta^\star$ simultaneously with high probability. We present a generic approach to constructing a frequentist CS using Bayesian tools, based on the fact that the ratio of a prior to the posterior at the ground truth is a martingale. We then present Hoeffding- and empirical-Bernstein-type time-uniform CSs and fixed-time confidence intervals for sampling WoR, which improve on previous bounds in the literature and explicitly quantify the benefit of WoR sampling.
 
\end{abstract}


\section{Introduction}

When data are collected sequentially rather than in a single batch with a fixed sample size, many classical statistical tools cannot naively be used to calculate uncertainty as more data become available. Doing so can quickly lead to overconfident and incorrect results (informally, ``peeking, $p$-hacking''). For these kinds of situations, the analyst would ideally have access to procedures that allow them to:
\begin{enumerate}[(a)]
	\itemsep0em 
	\item Efficiently calculate tight confidence intervals  whenever new data become available;
	\item Track the intervals, and use them to decide whether to continue sampling, or when to stop;
	\item Have valid confidence intervals (or $p$-values) at arbitrary data-dependent stopping times. 
\end{enumerate}
The desire for methods satisfying (a), (b), and (c) led to the development of \textit{confidence sequences}~(CS) --- sequences of confidence sets which are uniformly valid over a given time horizon $\calT$. Formally, a sequence of sets $\{ C_t \}_{t \in \calT}$ is a $(1-\alpha)$-CS for some parameter $\theta^\star$ if
\begin{equation}\label{eq::CS-defn}
\Pr(\forall t \in \calT,\ \theta^\star \in C_t ) \geq 1-\alpha ~~~ \equiv ~~~ \Pr(\exists t \in \calT : \theta^\star \notin C_t) \leq \alpha. 
\end{equation}
Critically, \eqref{eq::CS-defn} holds iff $\Pr(\theta^\star \notin C_\tau) \leq \alpha$ for arbitrary stopping times $\tau$ \cite{howard2018uniform}, yielding property~(c).
The foundations of CSs were laid by Robbins, Darling, Siegmund \& Lai \cite{darling1967confidence,lai1976confidence,robbins_boundary_1970,lai_boundary_1976}. The multi-armed bandit literature sometimes calls them `anytime' confidence intervals~\cite{jamieson_lil_2014,kaufmann2018mixture}. CSs have recently been developed for a variety of nonparametric problems \cite{howard2018uniform,wasserman2019universal,howard2019sequential}. 

This paper derives closed-form CSs when samples are drawn without replacement (WoR) from a finite population. The technical underpinnings are novel (super)martingales for both categorical (Section~\ref{section:discreteCategorical}) and continuous (Section~\ref{section:boundedRealValued}) observations. In the latter setting, our results unify and improve on the time-uniform with-replacement extensions of Hoeffding's \cite{hoeffding_probability_1963} and empirical Bernstein's inequalities by~\citet{maurer_empirical_2009}  that have been derived recently~\cite{howard_exponential_2018,howard2018uniform}, with several related inequalities for sampling WoR by Serfling~\cite{serfling1974probability} and extensions by \citet{bardenet2015concentration} and \citet{greene2017exponential}.

\begin{figure}
    \centering
    \includegraphics[width=0.7\textwidth]{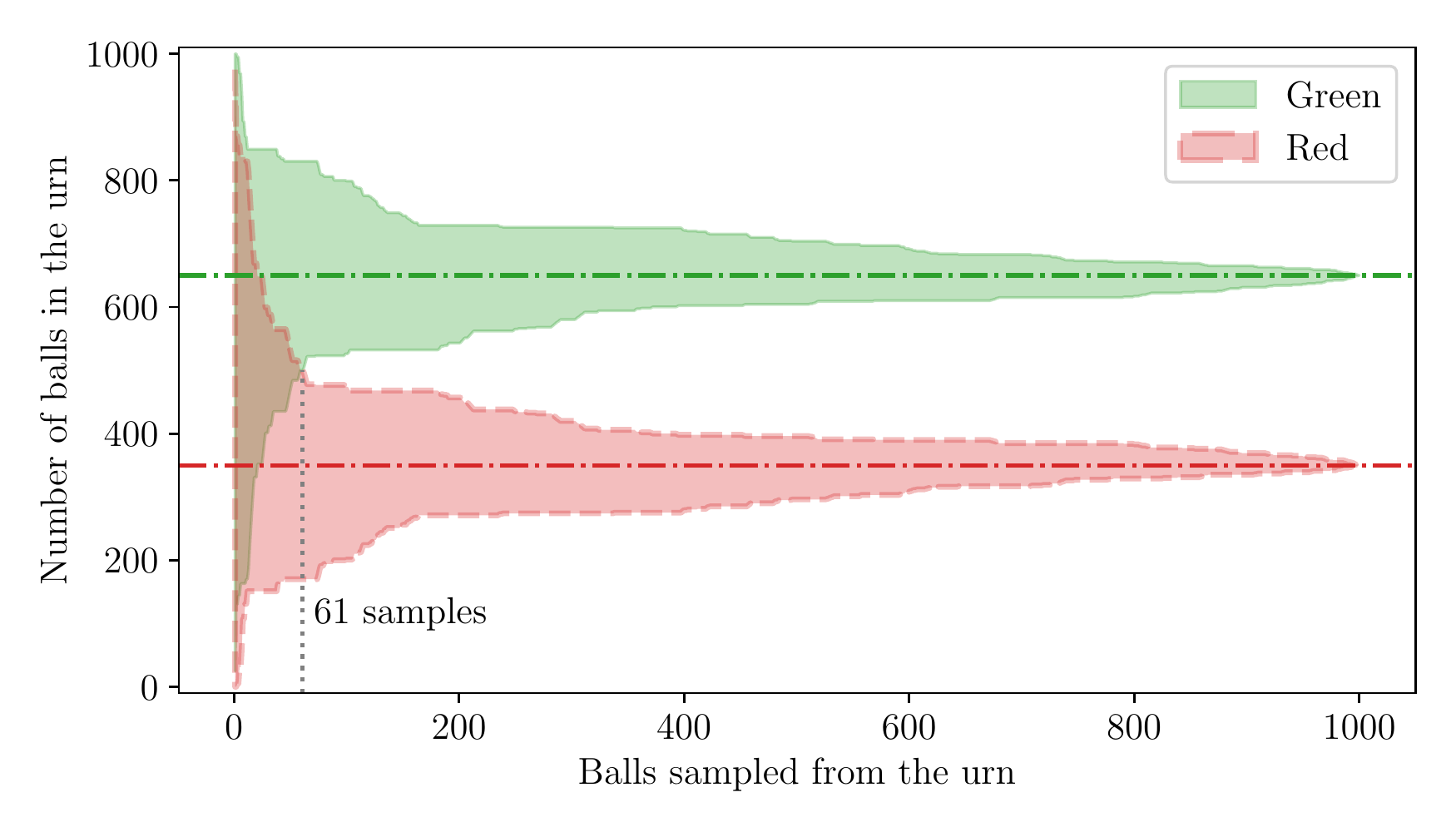}
    \caption[95\% CS for the number of green and red balls in an urn by sampling WoR. Notice that the true totals (650 green, 350 red) are captured by the CSs uniformly from the initial sample until all 1000 balls are observed. After sampling 61 balls in this particular example, the CSs cease to overlap, and we can conclude with 95\% confidence that there more green than red balls in the urn]{95\% CS for the number of green and red balls in an urn by sampling WoR\footnotemark. Notice that the true totals (650 green, 350 red) are captured by the CSs uniformly over time from the initial sample until all 1000 balls are observed. After sampling 61 balls in this example, the CSs cease to overlap, and we can conclude with 95\% confidence that there are more green than red balls in the urn.}
    \label{fig:CS}
\end{figure}
\footnotetext{Code to reproduce plots is available at \href{https://github.com/WannabeSmith/confseq_wor}{github.com/wannabesmith/confseq\_wor}.}

\textbf{Outline.}
In Section~\ref{section:discreteCategorical}, we use Bayesian ideas to obtain frequentist CSs for categorical observations. In Section~\ref{section:boundedRealValued}, we construct CSs for the mean of a finite set of bounded real numbers. We discuss implications for testing in Section~\ref{sec:testing}. Some prototypical applications are described in Appendix~\ref{section:fourMotivatingExamples}. The other appendices contain proofs, choices of tuning parameters, and computational considerations.

\subsection{Notation, supermartingales and the model for sampling WoR}
\label{section:observationModel}
Everywhere in this paper, the $N$ objects in the finite population $\{x_1,\dots,x_N\}$ are fixed and nonrandom. In the discrete setting (Section~\ref{section:discreteCategorical}) with $K\geq 2$ categories $\{c_k\}_{k=1}^K$, we have $x_i \in \{c_1, c_2,\dots, c_K\}$. In the continuous setting (Section~\ref{section:boundedRealValued}), $x_i \in [\ell, u]$ for some known bounds $\ell < u$. What is random is only the order of observation; the model for sampling uniformly at random WoR posits that
\begin{equation}\label{eq:sampling-wo-replacement}
X_t \mid \{X_1,\dots,X_{t-1}\} \sim \text{Uniform}(\{x_1,\dots,x_N\} \backslash \{X_1,\dots,X_{t-1}\}).
\end{equation}
All probabilities in this paper are to be understood as solely arising from observing fixed entities in a random order, with no distributional assumptions being made on the finite population. It is worth remarking on the power of this randomization---as demonstrated in our experiments, one can estimate the average of a deterministic set of numbers to high accuracy without observing a large fraction of the set.

The results in this paper draw from the theory of \textit{supermartingales}. While they can be defined in more generality, we provide a definition of supermartingales which will suffice for the theorems that follow.

A filtration is an increasing sequence of sigma fields.
For the entirety of this paper, we consider the `canonical' filtration $(\mathcal F_t)_{t =0}^N$ defined by $\mathcal F_t := \sigma(X_1, \dots, X_t)$, with $\mathcal F_0$ is the empty or trivial sigma field.
For any fixed $N \in \mathbb N$, a stochastic process $(M_t)_{t=0}^N$ is said to be a \textit{supermartingale} with respect to $(\mathcal F_t)_{t=0}^N$ if for all ${t \in \{0, 1, \dots, N-1\}}$, $M_t$ is measurable with respect to $\mathcal F_t$ (informally, $M_t$ is a function of $X_1,\dots,X_t$), and
\[ \E(M_{t+1} \mid \mathcal F_t) \leq M_t. \]
If the above inequality is replaced by an equality for all $t$, then $(M_t)_{t=0}^N$ is said to be a \textit{martingale}.

For succinctness, we use the notation $a_1^t := \{a_1,\dots,a_t\}$ and $[a]:=\{1,\dots,a\}$. Using this terminology, one can rewrite model \eqref{eq:sampling-wo-replacement} as positing that $X_t \mid \mathcal{F}_{t-1} \sim \text{Uniform}(x_1^N \backslash X_1^{t-1})$.

\section{Discrete categorical setting}
\label{section:discreteCategorical}

When observations are of this discrete form, the variables can be rewritten in such a way that they follow a hypergeometric distribution. In such a setting, the following ``prior-posterior-ratio martingale'' can be used to obtain CSs for parameters of the hypergeometric distribution which shrink to a single point after all data have been observed.

\subsection{The prior-posterior-ratio (PPR) martingale}

While the PPR martingale will be particularly useful for obtaining CSs when sampling discrete categorical random variables WoR from a finite population, it may be employed whenever one is able to compute a posterior distribution, and is certainly \emph{not limited to this paper's setting}. Moreover, this posterior distribution need not be computed in closed form, and computational techniques such as Markov Chain Monte Carlo may be employed when a conjugate prior is not available or desirable.

To avoid confusion, we emphasize that while we make use of terminology from Bayesian inference such as posteriors and conjugate priors, all of the probability statements with regards to CSs should be read in the frequentist sense, and are not interpreted as sequences of credible intervals.

Consider any family of distributions $\{F_\theta\}_{\theta \in \Theta}$ with density $f_\theta$ with respect to some underlying common measure (such as Lebesgue for continuous cases, counting measure for discrete cases).
Let $\theta^\star \in \Theta$ be a fixed parameter and let $\mathcal T = [N]$ where $N \in \mathbb N \cup \{\infty \}$. 
Suppose that $X_1 \sim f_{\theta^\star}(x)$ and
\[
X_{t+1} \sim f_{\theta^\star}(x \mid X_1^t) \quad \text{ for all $t \in \mathcal T$}.
\]
Let $\pi_0(\theta)$ be a prior distribution on $\Theta$, with posterior given by
\[ \pi_{t}(\theta) = \frac{\pi_0(\theta) f_{\theta}(X_1^t)}{\int_{\eta \in \Theta} \pi_0(\eta) f_\eta (X_1^t) d\eta}.\]
To prepare for the result that follows, define the \textit{prior-posterior ratio (PPR)} evaluated at $\theta \in \Theta$ as
\[ R_t(\theta) := \frac{\pi_0(\theta)}{\pi_t(\theta)}. \]

\begin{proposition}[Prior-posterior-ratio martingale]
\label{prop:priorPosteriorMartingale}
    For any prior $\pi_0$ on $\Theta$ that assigns nonzero mass everywhere, the sequence of prior-posterior ratios evaluated at the true $\theta^\star$, that is $(R_t(\theta^\star))_{t=0}^N$, 
    is a nonnegative martingale with respect to $(\mathcal{F}_t)_{t=0}^N$.
    Further, the sequence of sets  \[C_t := \{ \theta \in \Theta : R_t(\theta) < 1/\alpha \} \]
    forms a $(1-\alpha)$-CS for $\theta^\star$, meaning that $\Pr(\exists t \in \mathcal T: \theta^\star \notin C_t) \leq \alpha$.
\end{proposition}

The proof is given in Appendix \ref{proof:priorPosteriorMartingale}.

Going forward, we adopt the label \emph{working} before `prior' and `posterior' and encase them in `quotes' to emphasize that they constitute part of a Bayesian `working model', to contrast it against an assumed Bayesian model; the latter would be inappropriate given the discussion in Section~\ref{section:observationModel}. Next, we apply this result to the hypergeometric distribution. We will later examine the practical role of this working prior.

\subsection{CSs for binary settings using the hypergeometric distribution}
\label{section:binaryrvwithoutreplace}
Recall that a random variable $X$ has a hypergeometric distribution with parameters $(N, \HGP,n)$ if it represents the number of ``successes'' in $n$ random samples WoR from a population of size $N$ in which there are $\HGP$ such successes, and each observation is either a success or failure (1 or 0). The probability of a particular number of successes $x \in \{0, 1, \dots, \min(\HGP, n) \}$ is 
\[
{\small
\Pr(X = x) = \tbinom{N^+}{x} \tbinom{N-\HGP}{n - x} / \tbinom{N}{n}.}
\]
For notational simplicity, we consider the case when $n = 1$, that is we make one observation at a time, but this is not a necessary restriction. In fact, one would obtain the same CS at time ten if we repeatedly make one observation ten times, or make ten observations in one go. 
For a moment, let us view this problem from the Bayesian perspective, treating the fixed parameter $\HGP$ as a random parameter, which we call $\widetildeHGP$ to avoid confusion. We choose a beta-binomial `working prior' on $\widetildeHGP$ as it is conjugate to the hypergeometric distribution up to a shift in $\widetildeHGP$ \cite{fink1997compendium}. Concretely, suppose
\begin{align*}
\small
X_{t} \mid (\widetildeHGP, X_1, \dots, X_{t-1}) &\sim \HG \left(N - (t - 1), \widetildeHGP - \sum_{i=1}^{t-1} X_i, 1 \right),  \\
\widetildeHGP &\sim \BB(N, a, b),
\end{align*}
for some $a, b >0$. Then for any $t \in [N]$, the `working posterior' for $\widetildeHGP$ is given by 
\[
\small
\widetildeHGP - \sum_{i = 1}^t X_i \mid X_1^t  ~ \sim ~ \BB \Bigg(N-t, a + \sum_{i=1}^tX_i, b + t - \sum_{i=1}^t X_i \Bigg).
\]
Now that we have `prior' and `posterior' distributions for $\widetildeHGP$, an application of the prior-posterior martingale (Proposition~\ref{prop:priorPosteriorMartingale}) yields a CS for the true $\HGP$, summarized in the following theorem.

\begin{theorem}[CS for binary observations]
	\label{theorem:hyperGeoConfseq}
	Suppose $x_1^N \in \{0,1\}^N$ is a nonrandom set with the number of successes $\sum_{i=1}^N x_i \equiv \HGP$ fixed and unknown. Under observation model \eqref{eq:sampling-wo-replacement}, we have
	\[ X_{t} \mid X_1^{t-1} \sim \HG\Bigg (N - (t - 1), \HGP - \sum_{i=1}^{t-1} X_i, 1\Bigg).\]
For any beta-binomial `prior' $\pi_0$ for $\HGP$ with parameters $\smash{a,b>0}$ and induced `posterior' $\pi_t$,
\[C_t := \Bigg \{ n^+ \in [N] : \frac{\pi_0( n^+)}{\pi_t( n^+)} < \frac{1}{\alpha} \Bigg\}\]
is a $(1-\alpha)$-CS for $N^+$. Further, the running intersection, $( \bigcap_{s \leq t} C_t )_{t \in [N]}$ is also a valid $(1-\alpha)$-CS.

\end{theorem}
\begin{figure}[!htb]
	\centering \includegraphics[width=0.8\textwidth]{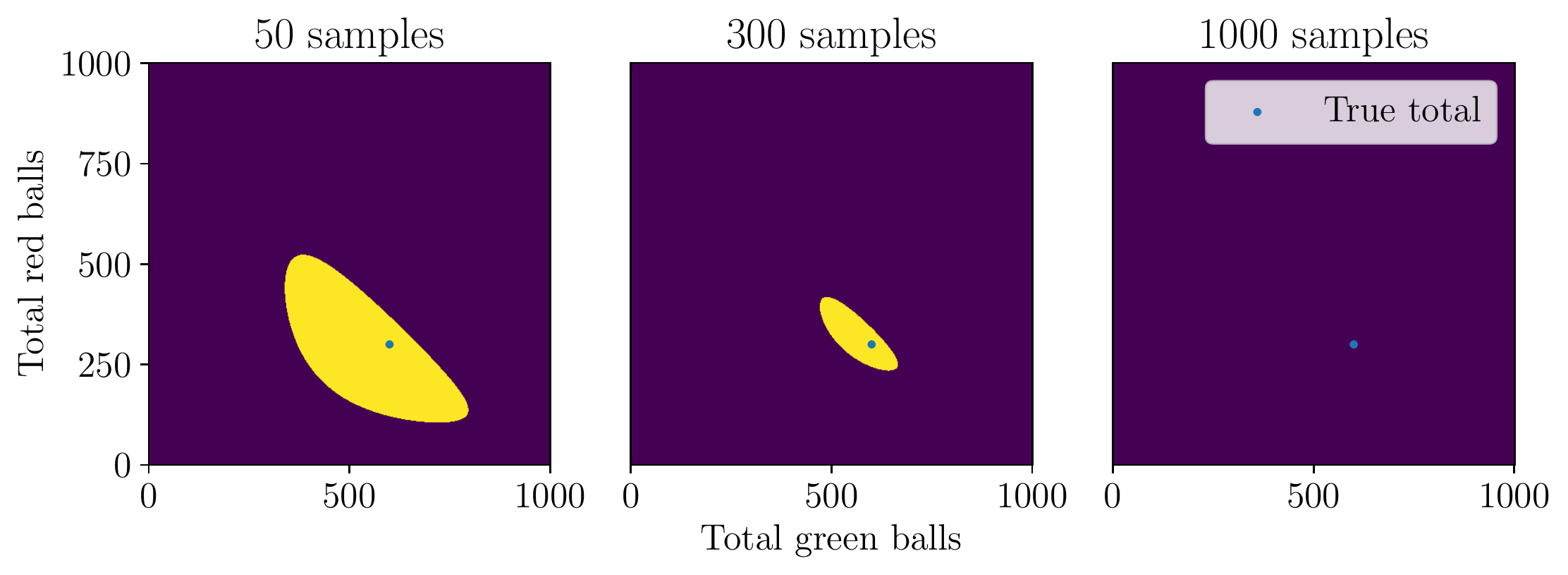}
	\caption{
    Consider sampling balls from an urn WoR with three distinct colors (red, green, and purple). In this example, the urn contains 1000 balls with 300 red, 600 green, and 100 purple. We only require a two-dimensional confidence sequence (yellow region) to capture uncertainty about all three totals. After around 300 balls have been sampled, we are quite confident that the urn is made up mostly of green; after 1000 samples, we know the totals for each color with certainty.
	}
	\label{fig:threePartyConfseq}
\end{figure}
The proof of Theorem \ref{theorem:hyperGeoConfseq} is a direct application of Proposition \ref{prop:priorPosteriorMartingale}. Note that for any `prior', the `posterior' at time $t = N$ is $\pi_N(n^+) = \mathds{1}(n^+ = N^+) $, so  $C_t$ shrinks to a point, containing only $N^+$. 
For $K > 2$ categories, Theorem~\ref{theorem:hyperGeoConfseq} can be extended to use a multivariate hypergeometric with a Dirichlet-multinomial prior to yield higher-dimensional CSs, but we leave the (notationally heavy) derivation to Appendix~\ref{section:priorPosteriorMultivariate}. See Figure~\ref{fig:threePartyConfseq} to get a sense of what these CSs can look like when $K=3$.

\subsection{Role of the `prior' in the prior-posterior CS}
\label{section:choiceOfPrior}

The prior-posterior CSs discussed thus far have valid (frequentist) coverage for any `prior' on $\HGP$, and in particular are valid for a beta-binomial `prior' with any data-independent choices of $a, b > 0$. 
Importantly, the corresponding CS always shrinks to zero width.
How, then, should the user pick $(a,b)$? 
Figure~\ref{fig:betaBinomial} provides some visual intuition.

\begin{figure}[!htb]
    \centering
        \includegraphics[width = 0.85 \textwidth]{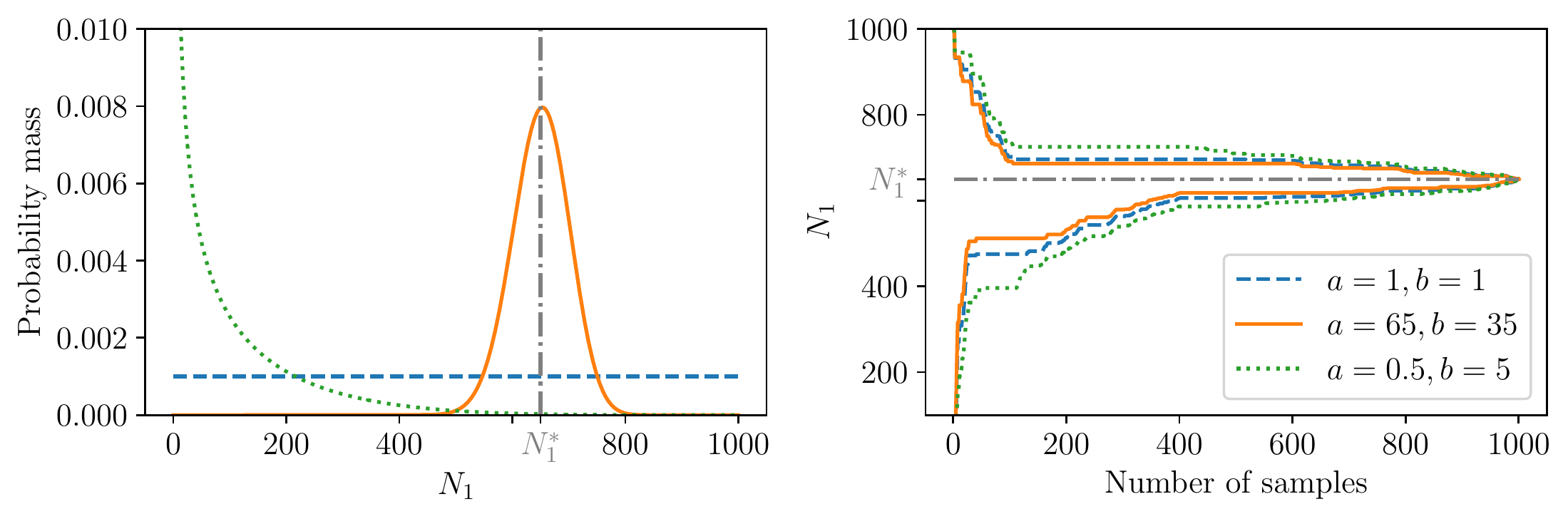}
  \caption{Beta-binomial probability mass function as a `prior' on $N_1^\star$ with different choices of ($a$, $b$), and the resulting PPR CS for the parameter $N_1^\star$ of a hypergeometric distribution when $(N_1^\star, N_2^\star) = (650, 350)$.}
  \label{fig:betaBinomial}
\end{figure}

These are our takeaway messages: (a) if the `prior' is very accurate (coincidentally peaked at the truth), the resulting CS is narrowest, (b) even if the `prior' is horribly inaccurate (placing almost no mass at the truth), the resulting CS is well-behaved and robust, albeit wider, (c) if we do not actually have any idea what the underlying truth might be, we suggest using a uniform `prior' to safely balance the two extremes. However, a more risky `prior' pays a relatively low statistical price.

\section{Bounded real-valued setting}
\label{section:boundedRealValued}

Suppose now that observations are real-valued and bounded as in Examples C and D of Appendix~\ref{section:fourMotivatingExamples}. Here we introduce Hoeffding- and empirical Bernstein-type inequalities for sampling WoR.

\subsection{Hoeffding-type bounds}

Recalling Section~\ref{section:observationModel}, we deal with a fixed batch $x_1^N$ of bounded real numbers $x_i \in [\ell,u]$ with mean $\mu := \frac{1}{N} \sum_{i=1}^N x_i$. 
Our CS for $\mu$ will utilize a novel WoR mean estimator,
\begin{equation}
\small
\widehat \mu_t := \frac{\sum_{i=1}^t X_i + \sum_{i=1}^{t}\frac{1}{N-i+1} \sum_{j=1}^{i-1} X_j}{t + \sum_{i=1}^{t}\frac{i-1}{N-i+1}}.
\label{eqn:withoutReplaceMeanEst}
\end{equation}
More generally, if $\lambda_1, \dots, \lambda_N$ is a predictable sequence (meaning $\lambda_t$ is $\mathcal F_{t-1}$-measurable for $t \in \{1, \dots, N\}$), then we may define the weighted WoR mean estimator,
\begin{equation}
\small
	\label{eqn:weightedWithoutReplaceMeanEst}
    \widehat \mu_t(\lambda_1^t) := \frac{\sum_{i=1}^t \lambda_i (X_i + \frac{1}{N-i+1} \sum_{j=1}^{i-1} X_j)}{\sum_{i=1}^t \lambda_i (1 + \frac{i-1}{N-i+1})},
\end{equation}
where it should be noted that if $\lambda_1 = \cdots = \lambda_N$ then $\widehat \mu_t(\lambda_1^t)$ recovers $\widehat \mu_t$.
Past WoR works \cite{serfling1974probability,bardenet2015concentration,greene2017exponential} base their bounds on the sample average $\sum_i X_i/t$.
Both $\widehat \mu_t$ and the sample average are conditionally biased and unconditionally unbiased (see Appendix~\ref{proof:hoeffdingWor} for more details).
As frequently encountered in Hoeffding-style inequalities for bounded random variables \cite{hoeffding_probability_1963}, define
\begin{equation}
\label{eqn:subGaussianCGF}
\small
\psi_{H}(\lambda) := \frac{\lambda^2 (u - \ell)^2}{8}.
\end{equation}
Setting $M^H_0 := 1$, we introduce a new exponential Hoeffding-type process for a predictable sequence $\lambda_1^N$,
				\begin{equation}
				\small
				\label{eqn:exponentialHoeffding}
				M_t^H := \exp \left \{ \sum_{i=1}^t \left [ \lambda_i \left ( X_i - \mu + \frac{1}{N-i+1}\sum_{j=1}^{i-1} (X_j - \mu) \right ) - \psi_H(\lambda_i)\right] \right \}.
				\end{equation}
\begin{theorem}[A time-uniform Hoeffding-type CS for sampling WoR]
\label{theorem:hoeffdingConfseq}
				Under the observation model and filtration $(\mathcal F_t)_{t=0}^N$ of Section~\ref{section:observationModel}, and for any predictable sequence $\lambda_1^N$, the process $(M_t^H)_{t=0}^N$ is a nonnegative supermartingale, and thus,

\[\Pr\left (\exists t \in [N] : \mu - \widehat \mu_t(\lambda_1^t) \geq \frac{\sum_{i=1}^t \psi_H(\lambda_i) + \log(1/\alpha)}{\sum_{i=1}^t \lambda_i \left (1 + \frac{i-1}{N-i+1}\right) } \right ) \leq \alpha. \]
Consequently,
\[ C^H_t := \widehat \mu_t(\lambda_1^t) \pm \frac{\sum_{i=1}^t \psi_{H}(\lambda_i) + \log(2/\alpha)}{\sum_{i=1}^t \lambda_i \left (1 + \frac{i-1}{N-i+1}\right )} 
~ \text{ ~ 	forms a $(1-\alpha)$-CS for $\mu$}.
\]

\end{theorem}
The proof in Appendix \ref{proof:hoeffdingWor} combines ideas from the with-replacement, \emph{time-uniform} extension of Hoeffding's inequality of Howard et al.~\cite{howard2018uniform,howard_exponential_2018} with the fixed-time, \textit{without-replacement} extension of Hoeffding's by Bardenet \& Maillard~\cite{bardenet2015concentration}, to yield a bound that improves on both. When $\lambda := \lambda_1 = \cdots = \lambda_N$ is a constant, the term 
\begin{equation}
\small    
A_t := \sum_{i=1}^{t}\frac{i-1}{N-i+1}
\end{equation}
 captures the `advantage' over the classical Hoeffding's inequality; we discuss this term more soon. 

In order to use the aforementioned CS, one needs to choose a predictable $\lambda$-sequence. First, consider the simpler case of a fixed real-valued $\lambda := \lambda_1 = \cdots \lambda_N$ as this will aid our intuition in choosing a more complex $\lambda$-sequence. In this case, $\lambda$ corresponds to a time $t_0 \in [N]$ for which the CS is tightest. If the user wishes to optimize the width of the CS for time $t_0$, then the corresponding $\lambda$ to be used is given by
\begin{equation}
	\label{eqn:lambdaOpt}
	\small
	\lambda := \sqrt{\frac{8 \log(2/\alpha)}{t_0 (u - \ell)^2}}.
\end{equation}
Alternatively, if the user does not wish to commit to a single time $t_0$, they can choose a $\lambda$-sequence akin to \eqref{eqn:lambdaOpt} but which spreads its width optimization over time. For example, one can use the sequence for $t \in \{1, \dots, N\}$,
\begin{equation}
\small
 \lambda_t := \sqrt{\frac{8 \log(2/\alpha)}{t \log(t+1) (u-\ell)^2}} \land \frac{1}{u-\ell},
\label{eq:lambdaSeqSpread}
 \end{equation}
where the minimum was taken to prevent the CS width from being dominated by early terms. Note however that any predictable $\lambda$-sequence yields a valid CS (see Appendix~\ref{section:choiceOfLambdaSequence} for more examples).

Optimizing a real-valued $\lambda = \lambda_1 = \cdots = \lambda_N$ for a particular time is in fact the typical strategy used to obtain the tightest fixed-time (i.e. non-sequential) Chernoff-based confidence intervals (CIs) such as those based on Hoeffding's inequality \cite{howard2018uniform, hoeffding_probability_1963}. This same strategy can be used with our WoR CSs to obtain tight fixed-time CIs for sampling WoR. Specifically, plugging \eqref{eqn:lambdaOpt} into Theorem~\ref{theorem:hoeffdingConfseq} for a fixed sample size $n \in [N]$, we obtain the following corollary.
\begin{corollary}[Hoeffding-type CI for sampling WoR]
\label{corollary:hoeffdingCI}
For any $n \in [N]$,
\begin{equation} \label{eq:HoeffdingCI}\widehat \mu_n \pm \frac{\sqrt{\frac{1}{2}(u-\ell)^2 \log(2/\alpha)}}{\sqrt{n} + A_n/\sqrt{n}} ~~\text{forms a $(1-\alpha)$ CI for $\mu$.}
\end{equation}
\end{corollary}
Notice that the classical Hoeffding confidence interval is recovered exactly, including constants, by dropping the $A_n$ term and using the usual sample mean estimator instead of $\widehat \mu_t$. To get a sense of how large the advantage is, note that 
\begin{align*}
\text{for small $n \ll N$, ~ } A_n &\asymp \sum_{i=1}^{n-1} i/N \asymp n^2/N,\\
\text{for large $n \approx N$, ~ } A_n &\asymp A_N = \sum_{i=1}^{N-1} \frac{i}{N-i} =  \sum_{j=1}^{N-1} \frac{N-j}{j} \asymp N\log N - (N-1).
\end{align*}
Thus, the advantage is negligible for $n = O(\sqrt{N})$, while it is substantial for $n=O(N)$, but it is clear that the CI of \eqref{eq:HoeffdingCI} is strictly tighter than Hoeffding's inequality for any $n$.

\subsection{Empirical Bernstein-type bounds}
\label{section:empiricalBernstein}
Hoeffding-type bounds like the one in Theorem~\ref{theorem:hoeffdingConfseq} only make use of the fact that observations are bounded, and they can be loose if only some observations are near the boundary of $[\ell,u]$ while the rest are concentrated near the middle of the interval. More formally, the CS of Theorem~\ref{theorem:hoeffdingConfseq} has the same width whether the underlying population $x_1^N$ has large or small variance $\sum_{i=1}^N (x_i - \mu)^2$---thus, they are tightest when the $x_i$s equal $\ell$ or $u$, and they are loosest when $x_i \approx (\ell+u)/2$ for all $i$. As an alternative that adaptively takes a variance-like term into account \cite{maurer_empirical_2009,balsubramani_sequential_2016}, we introduce a sequential, WoR, empirical Bernstein CS. 
As is typical in empirical Bernstein bounds \cite{howard2018uniform}, we use a different `subexponential'-type function,
\[ 
\small
\psi_{E}(\lambda) := (-\log (1-c\lambda) - c\lambda)/4 ~ \text{ ~ for any $\lambda \in [0, 1/c)$}
\] 
where $c := u- \ell$. $\psi_E$ seems quite different from $\psi_H$, but Taylor expanding $\log$ yields $\psi_E(\lambda)\approx c^2\lambda^2/8$. Indeed,
\begin{equation}\label{eq:psi-limit}
\small
\lim_{\lambda\to 0} \psi_{E}(\lambda) / \psi_H(\lambda) = 1.
\end{equation}
Note that one typically picks small $\lambda$, e.g.: set $t_0 = N/2, \ell=-1,u=1$ in \eqref{eqn:lambdaOpt} to get $\lambda_1 \propto 1/\sqrt{N}$. 

In what follows, we derive a time-uniform empirical-Bernstein inequality for sampling WoR. Similar to Theorem~\ref{theorem:hoeffdingConfseq}, underlying the bound is an exponential supermartingale. 
Set $M^E_0=1$, and recall that $c=u-\ell$ to define a novel exponential process for any $[0, 1/c)$-valued predictable sequence $\lambda_1, \dots \lambda_N$:
\begin{equation}
\label{eqn:exponentialEmpiricalBernstein}
\small
	M_t^E := \exp \left \{ \sum_{i=1}^t \left [ \lambda_i \left ( X_i - \mu + \frac{1}{N-i+1}\sum_{j=1}^{i-1} (X_j - \mu) \right ) - \left ( \frac{c}{2}\right )^{-2}(X_i - \widehat \mu_{i-1})^2\psi_E(\lambda_i)\right] \right \}.
\end{equation}

\begin{theorem}[A time-uniform empirical Bernstein-type CS for sampling WoR]
\label{theorem:empiricalBernsteinConfseq}
Under the observation model and filtration $(\mathcal F_t)_{t=0}^N$ of Section~\ref{section:observationModel}, and for any $[0, 1/c)$-valued predictable sequence $\lambda_1^N$, the process $(M_t^E)_{t=0}^N$ is a nonnegative supermartingale, and thus,
\[\Pr\left (\exists t \in [N] : \mu - \widehat \mu_t(\lambda_1^t) \geq \frac{\sum_{i=1}^t \left ( c/2 \right )^{-2}(X_i - \widehat \mu_{i-1})^2\psi_E(\lambda_i) + \log(1/\alpha)}{\sum_{i=1}^t \lambda_i \left (1 + \frac{i-1}{N-i+1}\right) } \right ) \leq \alpha. \]
Consequently,
\[ C^E_t := \widehat \mu_t(\lambda_1^t) \pm \frac{\sum_{i=1}^t \left ( c/2\right )^{-2}(X_i - \widehat \mu_{i-1})^2\psi_{E}(\lambda_i) + \log(2/\alpha)}{\sum_{i=1}^t \lambda_i \left (1 + \frac{i-1}{N-i+1}\right )} 
~ \text{ ~ 	forms a $(1-\alpha)$-CS for $\mu$}.
\]

\end{theorem}

The proof in Appendix \ref{proof:empiricalBernsteinWor} involves modifying the proof of Theorem 4 in Howard et al.~\cite{howard2018uniform} to use our WoR versions of $\widehat \mu_t$ and to include predictable values of $\lambda_t$.

\begin{figure}[!htb]
    \centering
    \includegraphics[width=0.9\textwidth]{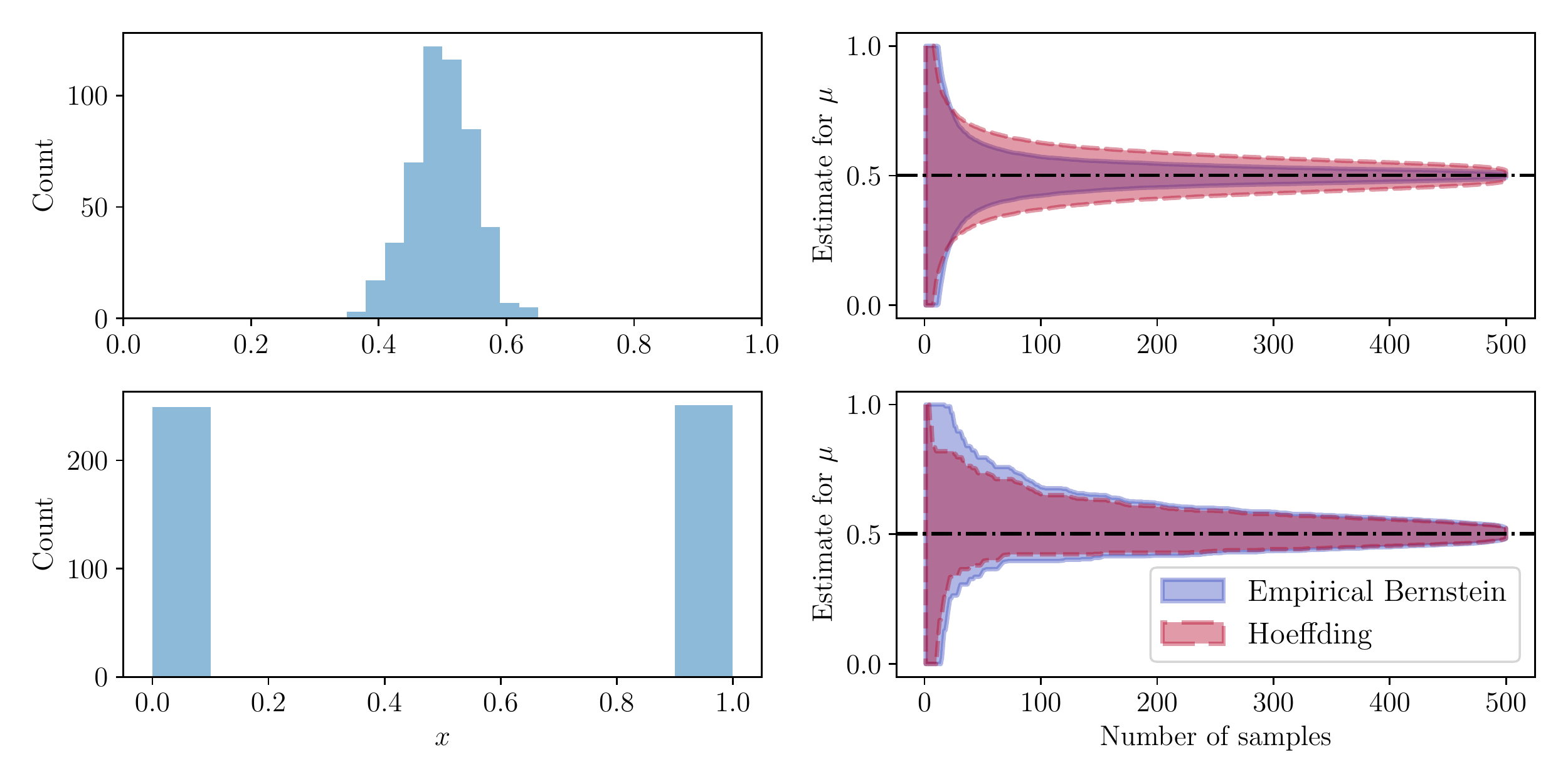}
    \caption{
    Left-most plots show the histogram of the underlying set of numbers $x_1^N\in [0,1]^N$, while right-most plots compare empirical Bernstein- and Hoeffding-type CSs for $\mu$. Specifically, the Hoeffding and empirical Bernstein CSs use the $\lambda$-sequences in \eqref{eq:lambdaSeqSpread} and \eqref{eq:empBernLambdaSeq}, respectively. As expected, in low-variance settings (top), $C^E_t$ is superior, but in a high-variance setting (bottom), $C^H_t$ has a slight edge.}
    \label{fig:hoeffdingBernsteinVariance}
\end{figure}

As before one must choose a $\lambda$-sequence to use $C^E_t$. We will again consider the case of a real-valued $\lambda := \lambda_1 = \cdots \lambda_N$ to help guide our intuition on choosing a more complex $\lambda$-sequence. Unlike earlier, we cannot optimize the width of $C_t^E$ in closed-form since $\psi_E$ is less analytically tractable. 
Once more, fact~\eqref{eq:psi-limit} comes to our rescue: substituting $\psi_H$ for $\psi_E$ and optimizing the width yields an expression like~\eqref{eqn:lambdaOpt}:
\begin{equation}
	\label{eqn:lambdaOpt2}
	\small
	\lambda^\star := \sqrt{\frac{2 \log(2/\alpha)}{\widehat V_t}},
\end{equation}
where $\widehat V_t := \sum_{i=1}^t (X_i - \widehat \mu_{i-1})^2$ is a variance process. 
However, we cannot use this choice of $\lambda^\star$ since it depends on $X_1^t$. Instead, we construct a predictable $\lambda$-sequence which mimics $\lambda^\star$ and adapts to the underlying variance as samples are collected. To heuristically optimize the CS for a particular time $t_0$, take an estimate $\widehat \sigma_{t-1}^2$ of the variance which only depends on $X_1^{t-1}$, and set
\begin{equation}
\label{eq:empBernLambdaSeq_t_0}
\small
 \lambda_t := \sqrt{\frac{2\log(2/\alpha)}{\widehat \sigma^2_{t-1} t_0}} \land \frac{1}{2c}.
\end{equation}
Alternatively, to spread the CS width optimization over time as in \eqref{eq:lambdaSeqSpread}, one can use the $\lambda$-sequence,
\begin{equation} 
\label{eq:empBernLambdaSeq}
\small
\lambda_t := \sqrt{\frac{2 \log(2/\alpha)}{\widehat \sigma_{t-1}^2t  \log(t+1)}} \land \frac{1}{2c},
\end{equation}
but again, any predictable sequence will suffice.

Similarly to the Hoeffding-type CS, we may instantiate the empirical Bernstein-type CS at a particular time to obtain tight CIs for sampling WoR. However, ensuring that the resulting fixed-time CI is valid when using a data-dependent $\lambda$-sequence requires some additional care. Suppose now that $X_1^n$ is a simple random sample WoR from the finite population, $x_1^N \in [\ell, u]^N$. If we randomly permute $X_1, \dots, X_n$ to obtain the sequence, $\widetilde X_1, \dots, \widetilde X_n$, we have recovered the observation model of Section~\ref{section:observationModel}, and thus Theorem~\ref{theorem:empiricalBernsteinConfseq} applies. We choose a $\lambda$-sequence which sequentially estimates the variance, but heuristically optimizes for the sample size $n$ as in \eqref{eq:empBernLambdaSeq_t_0}. For $t \in [n]$, define
\begin{equation}
\small
    \label{eq:lambdaEmpBernCI}
    \widetilde \lambda_t := \sqrt{\frac{2 \log(2/\alpha)}{ n \widetilde\sigma_{t-1}^2}} \land \frac{1}{2c} ~~~\text{where}~~~ \widetilde \sigma_t^2 := \frac{c^2/4 + \sum_{i=1}^t(\widetilde X_i - \widetilde \mu_i)^2}{t + 1} ~~~\text{and}~~~ \widetilde \mu_t := \frac{1}{t}\sum_{i=1}^t \widetilde X_i.
\end{equation} 
Here, an extra $c^2/4$ was added to $\widetilde \sigma^2_t$ so that it is defined at time 0, but this is simply a heuristic and any other choice of $\widetilde \sigma^2_0$ will suffice. The resulting CI can be summarized in the following corollary.
\begin{corollary}
\label{corollary:empBernCI}
    Let $X_1^n$ be a simple random sample WoR from the finite population $x_1^N$ and let $\widetilde X_1^n$ be a random permutation of $X_1^n$. Let $\widetilde \lambda_t$ be a predictable sequence such as the one in \eqref{eq:lambdaEmpBernCI} for each $t \in [n]$. Then for any $n \in [N]$,
    \[ \widehat \mu_n(\widetilde \lambda_1^n) \pm \frac{\sum_{i=1}^n (c/2)^{-2} (\widetilde X_i - \widetilde \mu_{i-1})^2 \psi_E(\widetilde \lambda_i) + \log(2/\alpha)}{\sum_{i=1}^n \widetilde \lambda_i\left ( 1 + \frac{i-1}{N-i+1} \right )} \text{ forms a $(1-\alpha)$ CI for $\mu$}.\]
\end{corollary}

The aforementioned CSs and CIs have a strong relationship with corresponding hypothesis tests. In the following section, we discuss how one can use the techniques developed here to sequentially test hypotheses about finite sets of nonrandom numbers.

\section{Testing hypotheses about finite sets of nonrandom numbers}
\label{sec:testing}

In classical hypothesis testing, one has access to i.i.d. data from some underlying distribution(s), and one wishes to test some property about them; this includes sequential tests dating back to Wald~\cite{wald1945sequential}. However, it is not often appreciated that it is possible to test hypotheses about a finite list of numbers that do not have any distribution attached to them. 
Recalling the setup of Section~\ref{section:observationModel}, this is the nonstandard setting we find ourselves in.
For instance in the same example as Figure~\ref{fig:CS}, we may wish to test:
\[
H_0: N_1^\star \leq 550 \quad \text{(At most 550 of the balls are green).}
\]
If we had access to each ball in advance, then we could accept or reject the null without any type-I or type-II error, but this is tedious, and so we sequentially take samples in a random order to test this hypothesis. The main question then is: \textit{how do we calculate a $p$-value $P_t$ that we can track over time, and stop sampling when $P_t \leq 0.05$?}

\begin{figure}[!htb]
    \centering
    \includegraphics[width=0.9\textwidth]{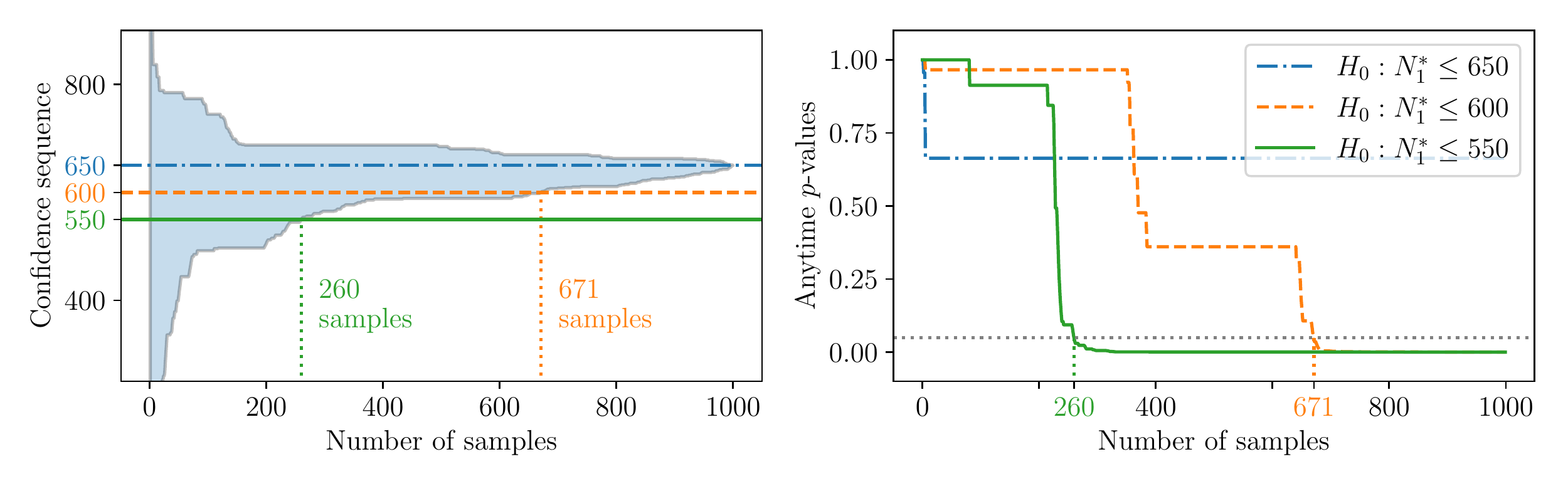}
    \caption{The duality between anytime $p$-values and CSs for three null hypotheses: $H_0: N_1^\star \leq D$ for $D \in \{550, 600, 650\}$. The first null is rejected at a $5\%$ significance level after 260 samples, exactly when the $95\%$ CS stops intersecting the null set $[0, 550]$. However, $H_0:N_1^\star \leq 650$ is never rejected since 650, the ground truth, is contained in the CS at all times from 0 to 1000.}
    \label{fig:CSversusPvalue}
\end{figure}

Luckily, we do not need any new tools for this, and our CSs provide a straightforward answer. Though we left it implicit, each confidence sequence $C_t$ is really a function of confidence level $\alpha$. Consider the family $\{C_t(q)\}_{q \in (0,1)}$ indexed by $q$, which we only instantiated at $q=\alpha$. Now, define
\begin{equation}\label{eq::anytime-p-value}
P_t := \inf\{ q: C_t(q) \cap H_0 = \emptyset \},
\end{equation}
which is the smallest error level $q$ at which $C_t(q)$ just excludes the null set $H_0$. This `duality' is familiar in non-sequential settings, and in our case it yields an anytime-valid $p$-value~\cite{johari2015always,howard2018uniform}, 
\[
\text{Under $H_0$, } \quad  \Pr(\exists t\in[N]: P_t \leq \alpha) \leq \alpha \text{ for any } \alpha \in [0,1].
\]
In words, if the null hypothesis is true, then $P_t$ will remain above $\alpha$ through the whole process, with probability $\geq 1-\alpha$. To more clearly bring out the duality to CSs, define the stopping time
\[
\tau := \inf\{t \in [N]: P_t \leq \alpha \}, \text{ and we set } \tau=N \text{ if the $\inf$ is not achieved}.
\]
Then under the null, $\tau=N$ (we never stop early) with probability $\geq 1-\alpha$. If we do stop early, then $\tau$ is exactly the time at which $C_t(\alpha)$ excluded the null set $H_0$. The manner in which anytime-valid $p$-values and CSs are connected through stopping times is demonstrated in Figure~\ref{fig:CSversusPvalue}.

In summary, our CSs directly yield $p$-values \eqref{eq::anytime-p-value} for composite null hypotheses. These $p$-values can be tracked, and are valid simultaneously at all times, including at arbitrary stopping times. Aforementioned type-I error probabilities are due to the randomness in the ordering, not in the data.

It is worth noting that our (super)martingales $(R_t)$, $(M^H_t)$ and $(M^E_t)$ also immediately yield `e-values'~\cite{shafer2019game} and hence `safe tests'~\cite{grunwald2019safe}, meaning that under nulls of the form in Figure~\ref{fig:CSversusPvalue}, they satisfy $\mathbb{E}M_{\tau} \leq 1$ for any stopping time $\tau$.

\section{Summary}\label{section:summary}
WoR sampling and inference naturally arise in a variety of applications such as finite-population studies and permutation-based statistical methods as outlined in Appendix~\ref{section:fourMotivatingExamples}. Furthermore, several machine learning tasks involve random samples from finite `populations', such as sampling (a) points for a stochastic gradient method, (b) covariates in a random order for coordinate descent, (c) columns of a matrix, or (d) edges in a graph.

In order to quantify uncertainty when sequentially sampling WoR from a finite set of objects, this paper developed three new confidence sequences: one in the discrete setting and two in the continuous setting (Hoeffding, empirical-Bernstein). Their construction was enabled by the development of new technical tools---the prior-posterior-ratio martingale, and two exponential supermartingales---which may be of independent interest. We clarified how these can be tuned (role of `prior' or $\lambda$-sequence), and demonstrated their advantages over naive sampling with replacement. Our CSs can be inverted to yield anytime-valid $p$-values to sequentially test arbitrary composite hypotheses. Importantly, these CSs can be efficiently updated, continuously monitored, and adaptively stopped without violating their uniform validity, thus merging theoretical rigor with practical flexibility.

\subsection*{Acknowledgements}
IW-S thanks Serge Aleshin-Guendel for conversations regarding Bayesian methods. AR thanks Steve Howard for early conversations. AR acknowledges funding from an Adobe Faculty Research Award, and an NSF DMS 1916320 grant.

\bibliographystyle{unsrtnat}
\bibliography{confseq_wo_replacement_NeurIPS}

\newpage

\appendix

\section{Four prototypical examples}
\label{section:fourMotivatingExamples}
The following examples are meant to demonstrate situations where we might care about sequentially quantifying uncertainty for parameters of finite populations (see Figure~\ref{fig:motivatingExamples}).

\subsection*{A. Opinion surveys (discrete categorical)}

Imagine you have access to a registry of phone numbers of a group of 1000 people, such as all residents of a neighborhood, voters in a township, or occupants of a university building. You wish to quickly determine the majority opinion on a categorical question, like preference of Biden vs. Trump. You pick names uniformly at random, call and ask. Obviously, you never call the same person twice. When can you confidently stop? In a typical run on a hypothetical ground truth of 650/350, our method stopped after 123 calls (Figure~\ref{fig:motivatingExamples}A). 

In the example of opinion surveys, the data are discrete and consist of 650 responses showing preference for Biden and 350 showing preference for Trump (encoded as ones and zeros, respectively). The observed data is thus a random permutation of 650 ones and 350 zeros. The CS used was the PPR CS for the hypergeometric distribution with a uniform `working prior' (i.e. $a = b =1$ in the beta-binomial pmf).

\subsection*{B. Permutation $p$-values  (discrete binary)}

Statistical inference is often performed using permutation tests. Formally, the permutation $p$-value is defined as $P_{\text{perm}} := \frac1{m!}\sum_{\pi \in S_m} I(T_m \geq T_{\pi(m)})$, where $T_m, T_{\pi(m)}$ are the original and permuted test statistics on $m$ datapoints, and $S_m$ is the set of all $m$-permutations (size $N=m!$). $P_{\text{perm}}$ is intractable to calculate for large $m$, so it is often approximated by randomly sampling $\pi$ with replacement (often $1000$ times, fixed and arbitrary). Instead, our tools allow a user to construct a CS for $P_{\text{perm}}$ and sequentially sample WoR until the CS is confident about whether $P_{\text{perm}}$ is below or above (say) $0.05$. In one example (small, so we can calculate $P_{\text{perm}}=0.04$ to verify accuracy), we stopped after 876 steps (Figure~\ref{fig:motivatingExamples}B).

The permutation test used in this example is a slight modification of the famous `Lady Tasting Tea' experiment \cite{fisher1956mathematics}. The experiment proceeds as follows.

There are 12 cups of tea with milk, half of which had the tea poured first, and the other half had milk poured first. The tea expert is told that half of the cups are milk-first and the other half are tea-first and is tasked with determining which ones are which. The null hypothesis is that the tea expert has no ability to distinguish between tea-first and milk-first (i.e. their guesses are independent of the order of milk/tea). Suppose they guess 10 out of 12 cups correctly. The statistical question becomes, ``what is the probability of guessing 10 or more cups correctly if the expert is guessing randomly?''. This probability is exactly the permutation $p$-value that the statistician is interested in.

To calculate this permutation $p$-value, we consider the set of all possible random guesses that the tea expert could have made, and compute the fraction of those which identify 10 or more cups correctly. If we randomly sample a sequence of possible guesses from the set of $\binom{12}{6}$ possible guesses and record whether 10 or more cups are correctly identified, then observations are a random stream of ones and zeros. We then construct a PPR CS with a uniform `working prior' for the number of ones, $N^+$ in this set to arrive at a CS for the permutation $p$-value, ${P := \tfrac{N^+}{\binom{12}{6}}}$.

\subsection*{C. Shapley values (bounded real-valued)}

First developed in game theory, Shapley values have been recently proposed as a measure of variable or data-point importance for supervised learning. Given a set of players $\{ 1, \dots, B \}$ and a reward function $\nu$, the Shapley value $\phi_b$ for player $b$ can be written as an average of $B!$ function evaluations, one for each permutation of $\{1, \dots, B\}$. As above, $\phi_b$ is intractable to compute and Monte-Carlo techniques are popular. This real-valued setting requires different CS techniques from the categorical setting. As Figure~\ref{fig:motivatingExamples}C unfolds from left to right (with $B=7$), it can be stopped adaptively with valid confidence bounds on all $\{\phi_b\}_{b=1}^B$. 
In this example, we consider a simple cost allocation problem. Suppose there are $n$ people that wish to share transportation to get from point A to their respective destinations, which are all in succession on the same street. Suppose that the cost of going from point A to the $i^\text{th}$ person's destination costs $c_i$, and without loss of generality suppose $c_1 < c_2 \cdots < c_n$. In this particular example, we used $n = 7$ with costs of 1, 10, 40, 80, 130, 175, and 200. The `cost', $\nu : 2^{[n]} \rightarrow \mathbb R$ of a trip is defined in the following natural way,
\begin{align*}
    &\nu(\emptyset) = 0\\
    &\nu(\{i\}) = c_i\\
    &\nu(S) = c_j \text{ where } c_j \geq c_k \text{ for all } k, j \in S
\end{align*}
The \textit{Shapley value}, $\phi_i$ for person $i$ can be written as,
\begin{equation}
    \phi_i = \frac{1}{n!} \sum_{\pi} \left [ \nu(S_{\pi, i} \cup \{i\}) - \nu(S_{\pi, i}) \right]
\end{equation}
where the sum is taken over all permutations $\pi$ of $[n]$, and $S_{\pi, i}$ is the set of numbers to the left of $i$ in the permutation $\pi([n])$.

Since the Shapley value $\phi_i$ is an average of $n!$ numbers, it may be tedious to compute for large $n$ especially when $\nu$ cannot be computed quickly. In our case, the summands have a crude upper bound of $c_n$ and a lower bound of 0 so we can randomly sample WoR from the set of permutations on $[n]$ to construct the empirical Bernstein CS of Theorem~\ref{theorem:empiricalBernsteinConfseq} with the $\lambda$-sequence of \eqref{eq:empBernLambdaSeq}. After 1252 permutations, we are able to conclude with high confidence which player has the highest Shapley value.


\subsection*{D. Tracking interventions (bounded real-valued)}

Suppose a state school board is interested in introducing a new program to help students improve their standardized testing skills. Before deploying it to each of their 3000 public schools, the board decides to incrementally introduce the program to randomly selected schools, measuring standardized test scores before and after its introduction. The board can construct a CS for the overall percentage increase in test scores (which could get worse), and stop the experiment once they are confident about the program's effectiveness. In Figure~\ref{fig:motivatingExamples}D, with effect size 20\%, the board can confidently decide to mandate the program statewide after 260 random schools have been trialed, but they may also continue tracking progress and stop later.
In this example, we simply generated 3000 observations from a Beta(3, 2) distribution, appropriately scaled to be between -100 and 100 (representing percentage changes in test scores). To construct a CS for the average change in test scores, we used the Hoeffding-type CS optimized for times 10, 100, and 1000. Note that this CS would be tighter if the empirical Bernstein CS were used as the Beta(3, 2) has a relatively small variance.

\begin{figure}[htb!]
	\centering \includegraphics[width=0.8\textwidth]{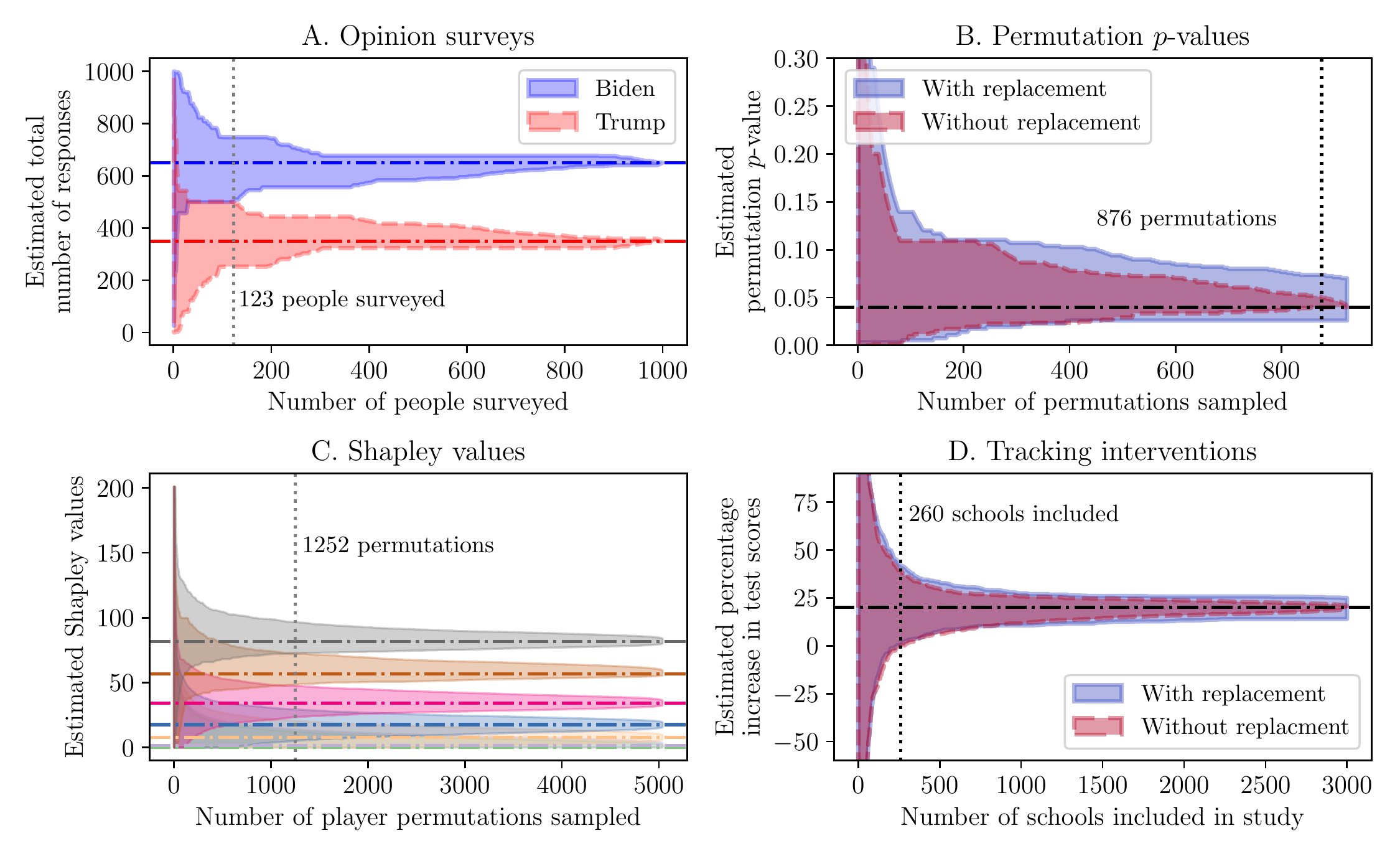}
	\caption{Typical simulation runs for the aforementioned examples, with more details in the Supplement. All experiments can be proactively monitored, optionally continued and adaptively stopped.}
	\label{fig:motivatingExamples}
\end{figure}

\section{Proofs of the main results}

\subsection{Proof of Proposition~\ref{prop:priorPosteriorMartingale}}
\label{proof:priorPosteriorMartingale}

The proof is broken into two steps. First, we prove that with respect to the filtration $(\mathcal F_t)_{t=0}^N$ outlined in Section~\ref{section:observationModel}, the prior-posterior ratio (PPR) evaluated at the true $\theta^\star \in \Theta$,
\begin{equation}
\label{eqn:priorPosteriorRatio}
    R_t(\theta^\star) := \frac{\pi_0(\theta^\star)}{\pi_t(\theta^\star)},
\end{equation}
is a nonnegative martingale with initial value one. Later, we invoke Ville's inequality~\cite{ville1939etude,howard_exponential_2018} for nonnegative supermartingales to construct the CS. 

\paragraph {Step 1.}
Let $\pi_0$ be any prior on $\Theta$ that assigns nonzero mass everywhere. Define the prior-posterior ratio, $R_t(\theta)$ as in \eqref{eqn:priorPosteriorRatio}. Writing the conditional expectation of $R_{t+1}(\theta^\star)$ given $X_1^t$ for any $t \in \{1, \dots, N\}$ in its integral form,
\begin{align*}
	\mathbb E(R_{t+1}(\theta^\star) \mid X_1^{t}) &= \int_{\mathcal X_{t+1}} \frac{\pi_0(\theta^\star)}{\pi_{t+1}(\theta^\star)} p_{\theta^\star}(x_{t+1} \mid X_1^t) dx_{t+1}\\
	&= \int_{\mathcal X_{t+1}} \frac{\pi_0(\theta^\star) \int_{\Theta} p_\eta (X_1^t, x_{t+1}) \pi_0(\eta) d\eta}{p_{\theta^\star}(X_1^t, x_{t+1}) \pi_0(\theta^\star)} p_{\theta^\star}(x_{t+1}) \mid X_1^t) dx_{t+1} & \text{(Bayes' rule)}\\
	&= \int_{\mathcal X_{t+1}} \frac{\pi_0(\theta^\star) \int_{\Theta} p_\eta (X_1^t, x_{t+1}) \pi_0(\eta) d\eta}{p_{\theta^\star}(X_1^t) \pi_0(\theta^\star)} dx_{t+1} & \text{(Bayes' rule again)}\\
    &= \int_{\mathcal X_{t+1}}\frac{\pi_0(\theta^\star)\int_\Theta p_\eta(X_1^t, x_{t+1}) \pi_0(\eta) d\eta}{\pi_t(\theta^\star) \int_\Theta p_\lambda(X_1^t) \pi_0(\lambda) d\lambda } dx_{t+1} & \text{(Bayes' rule again)} \\
    &= \frac{\pi_0(\theta^\star)}{\pi_t(\theta^\star)}\int_{\mathcal X_{t+1}}\frac{\int_\Theta p_\eta(X_1^t, x_{t+1}) \pi_0(\eta) d\eta}{\int_\Theta p_\lambda(X_1^t) \pi_0(\lambda)d\lambda} dx_{t+1} \\
    &= \frac{\pi_0(\theta^\star)}{\pi_t(\theta^\star)}\frac{\int_\Theta\int_{\mathcal X_{t+1}} p_\eta(X_1^t, x_{t+1}) dx_{t+1} \pi_0(\eta) d\eta}{\int_\Theta p_\lambda(X_1^t) \pi_0(\lambda)d\lambda} & \text{(Fubini's theorem)}\\
    &= \frac{\pi_0(\theta^\star)}{\pi_t(\theta^\star)}\frac{\int_\Theta p_\eta(X_1^t) \pi_0(\eta) d\eta}{\int_\Theta p_\lambda(X_1^t) \pi_0(\lambda)d\lambda} \quad = R_t(\theta^\star).\\
\end{align*}
Furthermore, for the case when $t=0$, 
\begin{align*}
    \mathbb E( R_1(\theta^\star)) &= \int_{\mathcal X_1} \frac{\pi_0(\theta^\star)\int_\Theta p_\eta(X_1) \pi_0(\eta) d\eta}{p_{\theta^\star}(X_1) \pi_0(\theta^\star)}p_{\theta^\star}(X_1)dx_1 \\
    &= \frac{\pi_0(\theta^\star)}{\pi_0(\theta^\star)} \int_{\mathcal X_1} \int_\Theta p_\eta(X_1) \pi_0(\eta) d\eta dx_1 & \text{(Bayes' rule)} \\
    &= \frac{\pi_0(\theta^\star)}{\pi_0(\theta^\star)}\int_\Theta \int_{\mathcal X_1} p_\eta(X_1)dx_1 \pi_0(\eta) d\eta & \text{(Fubini's theorem)}\\
    &= \frac{\pi_0(\theta^\star)}{\pi_0(\theta^\star)}\int_\Theta \pi_0(\eta) d\eta \quad = \frac{\pi_0(\theta^\star)}{\pi_0(\theta^\star)} = R_0 = 1.
\end{align*}

Establishing that $R_t(\theta^\star)$ is a nonnegative martingale with initial value one completes the first step.

\paragraph{Step 2.} 
Ville's inequality for nonnegative supermartingales \cite{ville1939etude, howard_exponential_2018} implies that for any $\beta > 0$,
\[ \Pr \left (\exists t \in [ N ] : R_t(\theta^\star) \geq \beta \right ) \leq \frac{\E (R_0(\theta^\star))}{\beta }. \]
In particular, for a threshold $\alpha \in (0, 1)$,
\begin{equation} 
\label{eqn:Ville}
\Pr\Big (\exists t \in [N] : R_t(\theta^\star) \geq 1/\alpha \Big ) \leq \alpha. 
\end{equation}
Define the sequence of sets for $t \in [N]$,
\[ C_t := \{ \theta : R_t(\theta) \leq 1/\alpha \}. \]
As a consequence of \eqref{eqn:Ville}, we have that
\[ \Pr\left (\forall t \in [N],\ \theta^\star \in C_t\right) \geq 1-\alpha, \]
as desired, which completes the proof.

\subsection{Proof of Theorem~\ref{theorem:hoeffdingConfseq}}
\label{proof:hoeffdingWor}

\begin{proof}
Similar to the proof of Proposition \ref{prop:priorPosteriorMartingale}, we proceed in two steps. First, we show that the exponential Hoeffding-type process \eqref{eqn:exponentialHoeffding} is a nonnegative supermartingale with respect to the filtration outlined in Section~\ref{section:observationModel}. We then apply Ville's inequality to this supermartingale and ultimately obtain the bound stated in the theorem.

We prove the bound for $[0, 1]$-bounded random variables but the general result holds by taking any $[\ell, u]$-bounded random variable, $X_i$ and applying the transformation, $X_i \mapsto (X_i - \ell) / (u-\ell)$

\paragraph{Step 1.} 

Let $(\Fcal_t)_{t=0}^N$ be the filtration defined in Section~\ref{section:observationModel}. Furthermore, let $\lambda_t \equiv \lambda_t(X_1, \dots, X_{t-1})$ be a sequence of $\Fcal_{t-1}$-measurable random variables. Consider the exponential Hoeffding-type process $(M_t^H)_{t=0}^N$ with a `predictable mixture',
\[ M_t^H := \exp \left \{ \sum_{i=1}^t \left [ \lambda_i \left ( X_i - \mu + Z_{i-1}^\star \right ) - \frac{\lambda_i^2}{8}\right] \right \} \equiv \prod_{i=1}^t \exp \left \{ \lambda_i \left ( X_i - \mu + Z_{i-1}^\star \right ) - \frac{\lambda_i^2}{8} \right \} \]
where $Z_i^\star = \frac{1}{N-i}\sum_{j=1}^i (X_j - \mu)$ and $M_0^H = 0$ by convention. Writing the conditional expectation of this process for any $t \geq 1$,
\begin{align*}
    \E(M_{t+1}^H \mid \Fcal_t) &= \E \left ( \prod_{i=1}^{t+1} \exp\left \{ \lambda_i (X_i - \mu + Z_{i-1}^\star) - \frac{\lambda_i^2}{8} \right \} \Bigm | \Fcal_{t}\right ) \\
    &= M_t^H \cdot \E\left ( \exp \left \{ \lambda_{t+1}(X_{t+1} - \mu + Z_{t}^\star) - \frac{\lambda_{t+1}^2}{8} \right \} \Bigm | \Fcal_t \right ).
\end{align*}
Using the fact that $\E(X_{t+1} - \mu + Z_t^\star \mid \Fcal_t) = 0$, the fact that $X_{t+1} \in [0, 1]$, and that $\lambda_{t+1}$ is $\Fcal_t$-measurable, we have by sub-Gaussianity of bounded random variables, 
\[\E\left ( \exp \left \{ \lambda_{t+1}(X_{t+1} - \mu + Z_{t}^\star)\right \} \Bigm | \Fcal_t \right ) \leq \exp\left \{ \frac{\lambda_{t+1}^2}{8} \right \}\]
and thus $\E(M_{t+1}^H \mid \Fcal_t) \leq M_t^H$. Therefore, with respect to the filtration $(\Fcal_t)_{t=0}^N$, we have that $M_t^H$ is a nonnegative supermartingale.

\paragraph{Step 2.} Now that we have shown that $M_t^H$ is a nonnegative supermartingale, we may apply Ville's inequality to obtain,
\[ \Pr \left(\exists t \in [N] : M_t^H \geq \frac{1}{\alpha} \right) \leq \alpha. \]
In particular, with probability at least $(1-\alpha)$, we have that for all $t \in [N]$, $M_t^H < \frac{1}{\alpha}$.

\paragraph{Step 3.} `Inverting' the above statement and solving for $\widehat \mu_t(\lambda_1^t) - \mu$, we get that with probability at least $(1-\alpha)$, for all $t \in [N]$,
\[ \widehat \mu_t(\lambda_1^t) - \mu < \frac{\sum_{i=1}^t \lambda_i^2 /8 + \log(1/\alpha)}{\sum_{i=1}^t \lambda_i \left ( 1 + \frac{i-1}{N-i+1} \right )}. \]
Applying all of the aforementioned logic to $-X_1, \dots, -X_t$ and $-\mu$, and taking a union bound, we arrive at the desired result, 
\[ \Pr \left ( \exists t \in [N] : | \widehat \mu_t(\lambda_1^t) - \mu| \geq  \frac{\sum_{i=1}^t \lambda_i^2 /8 + \log(2/\alpha)}{\sum_{i=1}^t \lambda_i \left ( 1 + \frac{i-1}{N-i+1} \right )} \right ) \leq \alpha, \]
which completes the proof.

\end{proof}

\textbf{Remark: $\widehat \mu_t$ is unconditionally unbiased.} Recalling the advantage term $A_t := \sum_{i=1}^{t}\tfrac{i-1}{N-i+1}$, a short calculation shows that $\widehat \mu_t$ \eqref{eqn:withoutReplaceMeanEst} has conditional expectation equaling a convex combination of $\widehat \mu_t,\mu$:
\[
\E[\widehat \mu_{t+1} | X_1^t ] = \frac{1 + A_{t+1} - A_t}{t + 1+ A_{t+1}} \mu + \frac{t + A_t}{t +1+ A_{t+1}} \widehat \mu_t.
\]
Multiplying both sides by $t+1+A_{t+1}$, we can write it in a recursive, telescoping form:
\[
\E[(t+1+A_{t+1})\widehat \mu_{t+1} | X_1^t ] = \mu + (A_{t+1}-A_t) \mu + (t+A_t) \widehat \mu_t.
\]
Taking expectation with respect to $X_t|X_1^{t-1}$, and using the above equation to evaluate the last term,
\[
\E[(t+1+A_{t+1})\widehat \mu_{t+1} | X_1^{t-1} ] = 2\mu + (A_{t+1}-A_{t-1}) \mu + (t-1+A_{t-1}) \widehat \mu_{t-1}.
\]
Unrolling this process out, we see that $\E[(t+1+A_{t+1})\widehat \mu_{t+1}] = (t+1)\mu + (A_{t+1} - A_0)\mu$. Since $A_0\equiv0$, we conclude that $\widehat \mu_{t+1}$ is an unconditionally unbiased estimator of $\mu$. 

Interestingly, the without-replacement mean estimator is not necessarily `consistent' (in the sense of recovering $\mu$ after all $N$ samples are drawn). However, the concept of consistency is subtle for finite populations as there is no longer any uncertainty after all samples are drawn. In any case, the without-replacement mean estimator was not introduced to replace the usual sample mean estimator in all without-replacement settings, but was simply the quantity that resulted from attempting to develop exponential supermartingales within this sample scheme. 
\subsection{Proof of Theorem~\ref{theorem:empiricalBernsteinConfseq}}
\label{proof:empiricalBernsteinWor}

\begin{proof}
    Much like the proof of Theorem \ref{theorem:hoeffdingConfseq}, the proof proceeds in three steps: (1) showing that an exponential empirical Bernstein-type process is a supermartingale, (2) applying Ville's inequality, and (3) inverting the process and taking a union bound. Again, we prove the result for $[0,1]$-bounded random variables since for an $[\ell, u]$-bounded random variable $X_i$, one can make the transformation $X_i \mapsto (X_i - \ell) / (u-\ell)$
    \paragraph{Step 1.} Let $(\Fcal_t)_{t=0}^N$ be the filtration defined in Section~\ref{section:observationModel}. Let $\lambda_t \equiv \lambda_t(X_1, \dots, X_{t-1})$ be a sequence of $\Fcal_{t-1}$-measurable random variables. Consider the exponential empirical Bernstein-type process, $(M_t^E)_{t=0}^N$ with a `predictable mixture',
    \begin{align*}
        M_t^E &:= \exp\left \{ \sum_{i=1}^t \left [ \lambda_i\left (X_i - \mu + Z_{i-1}^\star \right) - 4(X_i - \widehat \mu_{i-1})^2 \psi_E(\lambda_i) \right ] \right \}\\
        &\equiv \prod_{i=1}^t \exp\left \{  \lambda_i\left (X_i - \mu + Z_{i-1}^\star\right ) - 4(X_i - \widehat \mu_{i-1})^2 \psi_E(\lambda_i) \right \}
    \end{align*}
    where $M_0^E := 0$. Writing out the conditional expectation of $M_{t+1}^E$ given $\Fcal_t$ for $t \in [N]$,
    \begin{align*} 
    \E \left ( M_{t+1}^E \mid \Fcal_t\right ) &= M_t^E \cdot \E \left ( \exp\left \{ \lambda_{t+1}\left (X_{t+1} - \mu + Z_{t}^\star \right) - 4\psi_E(\lambda_{t+1}) \left ( X_{t+1} - \widehat \mu_{t} \right )^2 \right \} \Bigm | \Fcal_t \right ).
    \end{align*}
    Therefore, it suffices to show that for any $t \in [N]$,
    \[ \E \left ( \exp\left \{ \lambda_{t+1}\left (X_{t+1} - \mu + Z_{t}^\star \right) - 4\psi_E(\lambda_{t+1}) \left ( X_{t+1} - \widehat \mu_{t} \right )^2 \right \} \Bigm | \Fcal_t \right ) \leq 1.\]
    For succinctness, denote 
	\[Y_{t+1} := X_{t+1} + \frac{1}{N-t} \sum_{j=1}^t X_j - \frac{N}{N-t}\mu ~~~\text{ and }~~~ \delta_t := \widehat \mu_t + \frac{1}{N-t}\sum_{j=1}^t X_j - \frac{N}{N-t}\mu. \]
	Note that $Y_{t+1}$ is conditionally mean zero.
	It then suffices to prove that for any $(0, 1)$-bounded, $\Fcal_t$- measurable $\lambda_{t+1} \equiv \lambda_{t+1}(X_1, \dots, X_t)$, 
	\[ \E \Bigg [\exp \Bigg \{ \lambda_{t+1} Y_{t+1} - 4(Y_{t+1} - \delta_t)^2\psi_{E}(\lambda_{t+1}) \Bigg \}  \Bigm | \Fcal_t \Bigg] \leq 1.\]
	Indeed, in the proof of Proposition 4.1 in Fan et al. \cite{fan2015exponential}, $\exp\{\xi\lambda  - 4\xi^2 \psi_{E} (\lambda)\} \leq 1 + \xi \lambda$ for any $\lambda \in [0, 1)$ and $\xi \geq -1$. Setting $\xi := Y_{t+1} - \delta_t = X_{t+1} - \widehat \mu_{t}$,
	\begin{align*}
		&\E \Bigg [\exp \Bigg \{ \lambda_{t+1} Y_{t+1} - 4(Y_{t+1} - \delta_t)^2 \psi_{E}(\lambda_{t+1} ) \Bigg \}  \Bigm | \Fcal_t \Bigg]\\
		&=\E \Big [\exp \Big \{ \lambda_{t+1} (Y_{t+1}-\delta_t) - 4(Y_{t+1} - \delta_t)^2 \psi_{E}(\lambda_{t+1}) \Big \}  \bigm | \Fcal_t \Big] \exp(\lambda_{t+1} \delta_t)\\
		    &\leq \E\Big [1 + (Y_{t+1} - \delta_t )\lambda_{t+1} \mid \Fcal_t \Big ]\exp(\lambda_{t+1} \delta_t) \quad \stackrel{(i)}{=} \E\big [1 - \delta_t \lambda_{t+1} \mid \Fcal_t \big ]\exp(\lambda_{t+1} \delta_t) ~ \quad 
		       \stackrel{(ii)}{\leq} 1,
	\end{align*}
	where equality $(i)$ follows from the fact that $Y_{t+1}$ is conditionally mean zero as mentioned earlier,
	and inequality $(ii)$ follows from the inequality $1-x \leq \exp(-x)$ for all $x \in \mathbb R$.
	
\paragraph{Step 2.} Now that we have established that $M_t^E$ is a nonnegative supermartingale, we apply Ville's inequality to obtain,
\[ \Pr\left ( \exists t \in [N] : M_t^E \geq \frac{1}{\alpha} \right ) \leq \alpha. \]
\paragraph{Step 3.} Solving for $\widehat \mu_t - \mu$ in the inequality in the above probability statement, we get that
\[ \Pr\left ( \exists t \in [N] :  \widehat \mu_t - \mu \geq \frac{\sum_{i=1}^t 4\psi_E(\lambda_i)(X_i - \widehat \mu_{i-1})^2 + \log(1/\alpha)}{\sum_{i=1}^t \lambda_i \left ( 1 + \frac{i-1}{N-i+1} \right )} \right ) \leq \alpha. \]
Applying the same logic to $-X_1, \dots, -X_t$ and $-\mu$, and taking a union bound, we arrive at the desired result,
\[ \Pr\left ( \exists t \in [N] :  |\widehat \mu_t - \mu| \geq \frac{\sum_{i=1}^t 4\psi_E(\lambda_i)(X_i - \widehat \mu_{i-1})^2 + \log(2/\alpha)}{\sum_{i=1}^t \lambda_i \left ( 1 + \frac{i-1}{N-i+1} \right )} \right ) \leq \alpha.\]
\end{proof}

\section{Sampling multivariate binary variables WoR}
\label{section:priorPosteriorMultivariate}

The prior-posterior martingale from Section~\ref{section:binaryrvwithoutreplace} extends naturally to the multivariate case as follows. Suppose we have $N$ objects, each belonging to one of $K \geq 2$ categories, and there are $N_1^\star, \dots, N_K^\star$ objects from each category, respectively. Let $c$ denote the category of a randomly sampled object, and let
\[ \mathbf{X} := \begin{pmatrix} \mathds{1}(c = 1) & \mathds{1}(c=2) &\cdots & \mathds{1}(c=K) \end{pmatrix}.\] Then $\mathbf{X}$ is said to follow a multivariate hypergeometric distribution with parameters $N$, $(N_1^\star, \dots, N_K^\star)$, and $n = 1$ and has probability mass function,
\[ \Pr(\mathbf{X} = x) = \frac{\prod_{k=1}^K \binom{N_k^\star}{x_k}}{\binom{N}{n}}.\]
Note that $\sum_{k=1}^K x_k = 1$ and $x_k \in \{0, 1\}$ for each $k \in \{1, \dots, K\}$. More generally, if $n \geq 2$ objects are sampled WoR, then $\mathbf{X}$ would have the same probability mass function with $x_1, \dots, x_K \in \{1, \dots, n\}$ such that $\sum_{k=1}^K x_k = n$. As in Section~\ref{section:binaryrvwithoutreplace}, we will consider the case where $n=1$ for notational simplicity.

Let us now view this random variable and the fixed multivariate parameter $\MHGP := (N_1^\star, \dots, N_K^\star)$ from the Bayesian perspective as in Section~\ref{section:binaryrvwithoutreplace} by treating $\MHGP$ as a random variable which we denote by $\widetilde{\mathbf{N}}^\star$ to avoid confusion. Suppose that
\[ \mathbf{X}_t \mid (\widetilde{\mathbf{ N}}^\star, \mathbf{X}_1, \dots, \mathbf{X}_{t-1}) \sim \MHG \Bigg (N - (t-1), \widetilde{\mathbf{N}}^\star - \sum_{i=1}^{t-1} \mathbf{X}_i, 1 \Bigg), \quad \text{and}\]
\[ \widetilde{\mathbf{N}}^\star \sim \DM(N, \mathbf{a}) \]
for some $\mathbf{a}:= (a_1, \dots, a_K)$ with $a_k > 0$ for each $k \in \{ 1, \dots, K \}$. Then for any $t \in \{ 1, 2, \dots, N \}$,
\[ \widetilde{\mathbf{N}}^\star - \sum_{i=1}^t \mathbf{X}_i \mid (\mathbf{X}_1, \dots, \mathbf{X}_t) \sim \DM \Bigg (N-t, \mathbf{a} + \sum_{i=1}^t \mathbf{X}_i \Bigg). \]
With these prior and posterior distributions, we're ready to invoke Proposition~\ref{prop:priorPosteriorMartingale} to obtain a sequence of confidence sets for $\mathbf{N}^\star$.

\begin{theorem}[Confidence sequences for multivariate hypergeometric parameters]
	\label{theorem:multHyperGeoConfseq}
Suppose that 

\[\mathbf{X}_t \mid (\mathbf{X}_1, \dots, \mathbf{X}_{t-1}) \sim \MHG \Bigg (N - (t-1), \MHGP - \sum_{i=1}^{t-1} \mathbf{X}_i, 1 \Bigg).\]

Let $\pi_0$ and $\pi_t$ be the Dirichlet-multinomial prior with positive parameters $\mathbf{a} = (a_1, \dots, a_K)$ and corresponding posterior, $\pi_t$, respectively. Then the sequence of sets $(C_t)_{t=0}^N$ defined by
\[C_t := \Bigg \{ \mathbf{n} \in \{0, \dots, N\}^{K} : \sum_{k=1}^K \mathbf{n}_k = N \text{ and } \frac{\pi_0( \mathbf{n})}{\pi_t( \mathbf{n})} < \frac{1}{\alpha} \Bigg\}\] is a $(1-\alpha)$-CS for $\MHGP$. Furthermore, the running intersection, $\bigcap_{s \leq t} C_t$ is a $(1-\alpha)$-CS for $\MHGP$.

\end{theorem}

\begin{proof}
This is a direct consequence of Theorem \ref{prop:priorPosteriorMartingale} applied to the multivariate hypergeometric distribution with a Dirichlet-multinomial prior.
\end{proof}

\section{Coupling the `prior' with the stopping rule to improve power}\label{section::additional-sims}

Somewhat at odds with their intended use-case, working `priors' need not always be chosen to reflect the user's prior information. When approximating $p$-values for permutation tests, for example, it is of primary interest to conclude whether $P_\text{perm}$ is above or below some prespecified $\alpha_\text{perm} \in (0,1)$ with high confidence as quickly as possible. As discussed in Theorem~\ref{theorem:hyperGeoConfseq}, the CS for $P_\text{perm}$ will shrink to a single point regardless of the prior, so if $P_\text{perm}$ is much larger or much smaller than $\alpha_\text{perm}$, we expect to discover the decision rule, ``reject'' versus ``do not reject'' rather quickly. It is when $P_\text{perm}$ is very close to $\alpha_\text{perm}$ that the user desires sharper confidence intervals, so that they can make decisions sooner (see Figure \ref{fig:uniformStrongPriorPermutation}). In this case, they simply need to place more mass near the decision boundary, with a necessary tradeoff between the sharpness of confidence sets near $\alpha_\text{perm}$ and the size of the neighborhood around $\alpha_\text{perm}$ for which this sharpness is realized. 

\begin{figure}[!htb]
	\includegraphics[width=1\textwidth]{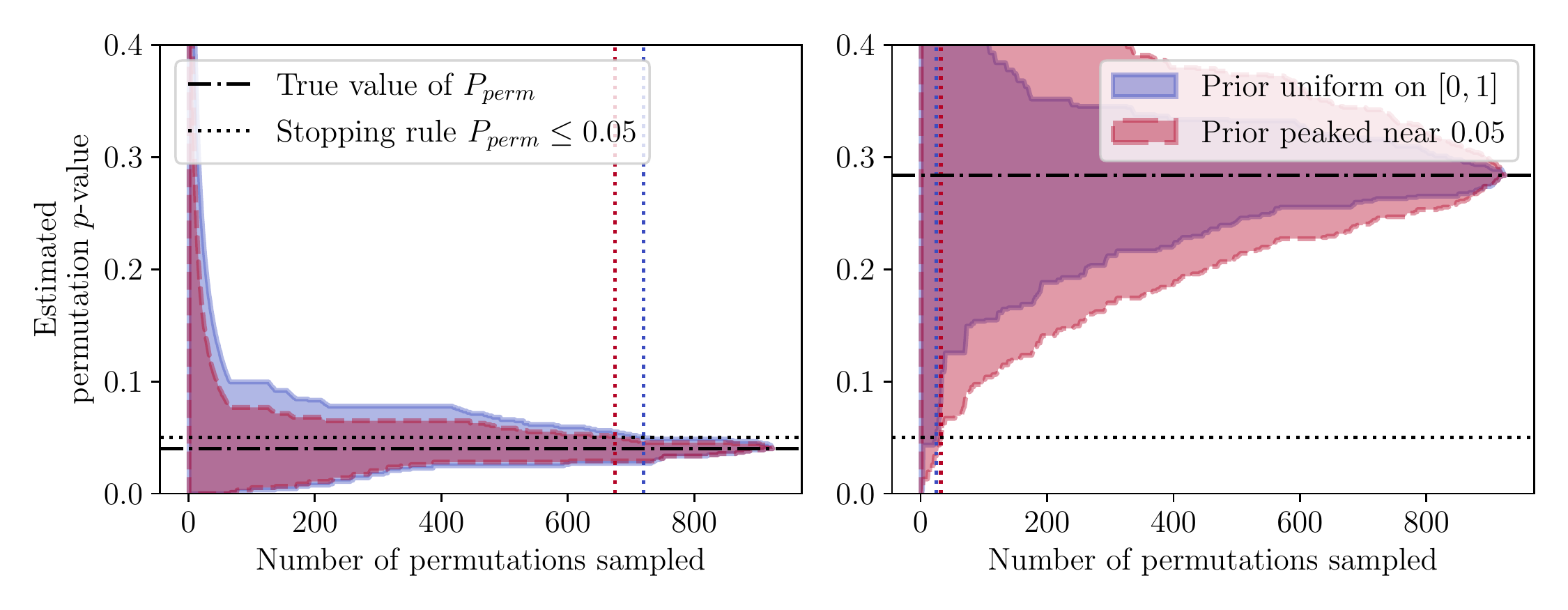}
	\caption{Comparing priors for Example B: using a uniform prior versus a prior peaked near 0.05. When the decision rule is to stop whenever the CS is entirely on one side of 0.05, coupling the prior to the decision rule leads to earlier stopping.}
	\label{fig:uniformStrongPriorPermutation}
\end{figure}

\section{Choosing a \texorpdfstring{$\lambda$}{lambda}-sequence for Hoeffding and empirical Bernstein CSs}
\label{section:choiceOfLambdaSequence}
Recall the Hoeffding-type CS of Theorem~\ref{theorem:hoeffdingConfseq}, 
\[ C_t^H := \widehat \mu_t(\lambda_1^t) \pm \underbrace{\frac{\sum_{i=1}^t \psi_H(\lambda_i) + \log(2/\alpha)}{\sum_{i}\lambda_i \left ( 1 + \frac{i-1}{N-i+1} \right )}}_{\text{width } W_t}\]
In Section~\ref{section:boundedRealValued}, we presented the $\lambda$-sequence,
\begin{equation}
\label{eq:lambdaSequenceHoeffding}
    \lambda_t := \sqrt{\frac{8 \log (2/\alpha)}{t \log (t+1) (u - \ell)^2}} \land \frac{1}{u-\ell}.
\end{equation}
This is visually similar to the single value of $\lambda \in \mathbb R$,
\[ \lambda := \sqrt{\frac{8 \log(2/\alpha)}{t_0 (u-\ell)^2}} \]
which optimizes the bound for time $t_0$. Two natural questions arise: (1) where did the extra $\log(t)$ in \eqref{eq:lambdaSequenceHoeffding} come from, and (2) why this particular $\lambda$-sequence and not others? The answers to these questions are based on some heuristics derived by \citet{waudby2020variance} in the with-replacement setting. To make matters simpler, ignore the $\left (1 + \frac{i-1}{N-i+1} \right)$ term in the CS and consider the scaling of the width $W_t$,
\[ W_t \asymp \frac{\sum_{i=1}^t\psi_H(\lambda_i)}{\sum_{i=1}^t \lambda_i}  \asymp \frac{\sum_{i=1}^t \lambda_i^2}{\sum_{i=1}^t \lambda_i}.\]
When the method of mixtures is used to obtain CSs in the with-replacement setting, their widths often follow a $\sqrt{\log t / t}$ rate \cite{howard2018uniform}. Following the approximations in Table~\ref{tab:schedules}, we may opt to pick a sequence $(\lambda_i)_{i=1}^\infty$ which scales like $1/\sqrt{i \log i}$ to obtain a width $W_t \asymp \sqrt{\log t / t}$. In particular, scaling $\lambda_i$ as $1/\sqrt{i \log i}$ is simply an effort to obtain CSs with reasonable widths. The same arguments combined with \eqref{eq:psi-limit} can be applied to the empirical Bernstein CS to obtain \eqref{eq:empBernLambdaSeq}.

Furthermore, we truncate the $\lambda$-sequence in \ref{eq:lambdaSequenceHoeffding} to prevent the CS width from being dominated by large $\lambda_t$ at small $t$. It is important to keep in mind that \textit{any} sequence would have yielded a valid CS. The choice presented here was derived based on a heuristic argument and kept because of its reasonable empirical performance.

\begin{table}[!htbp]
\begin{center}
 \begin{tabular}{|c c c c|} 
 \hline
 Sequence $(\lambda_i)_{i=1}^\infty$ & $\sum_{i=1}^t \lambda_i$ & $\sum_{i=1}^t \lambda_i^2$ & Width $W_t$ \\ [0.5ex] 
 \hline
 \hline
 $\asymp 1/i $ & $\asymp \log t$ & $\asymp 1$ & $1/\log t$ \\
 \hline
 $\asymp \sqrt{\log i/i}$ & $\asymp \sqrt{t \log t}$ & $\asymp \log^2 t$ & $\asymp \log^{3/2} t / \sqrt{t}$ \\
 \hline
 $\asymp 1/{\sqrt i}$ & $\asymp \sqrt{t}$ & $\asymp\log t$ & $\asymp \log t / \sqrt{t}$ \\ 
 \hline
 $\asymp 1/\sqrt{i \log i} $ & $\asymp \sqrt{t/\log t}$  & $\asymp \log\log t$ & $\asymp \sqrt{\log t/t}$ \\
 \hline
 $\asymp 1/\sqrt{i \log i \log\log i} $ & $\asymp \sqrt{t/\log t}$  & $\asymp \log\log\log t$ & $\asymp \sqrt{\log t/t}$ \\
 \hline
\end{tabular}
\caption{Above, we think of $\log x$ as $1 \vee \log(1 \vee x)$ to avoid trivialities.
The claimed rates are easily checked by approximating the sums as integrals, and taking derivatives. For example, $\tfrac{d}{dx}\log \log x = 1/x \log x$, so the
sum of $\sum_{i \leq t} 1/i \log i \asymp \log\log t$. It is worth remarking that for $t=10^{80}$, the number of atoms in the universe, $\log \log t \approx 5.2$, which is why we treat $\log \log t$ as a constant when expressing the rate for $W_t$. The iterated logarithm pattern in the the last two lines of the table can be continued indefinitely.}\label{tab:schedules}
\end{center}
 \end{table}
 
 \begin{figure}[!htbp]
     \centering
     \includegraphics[width=0.7\textwidth]{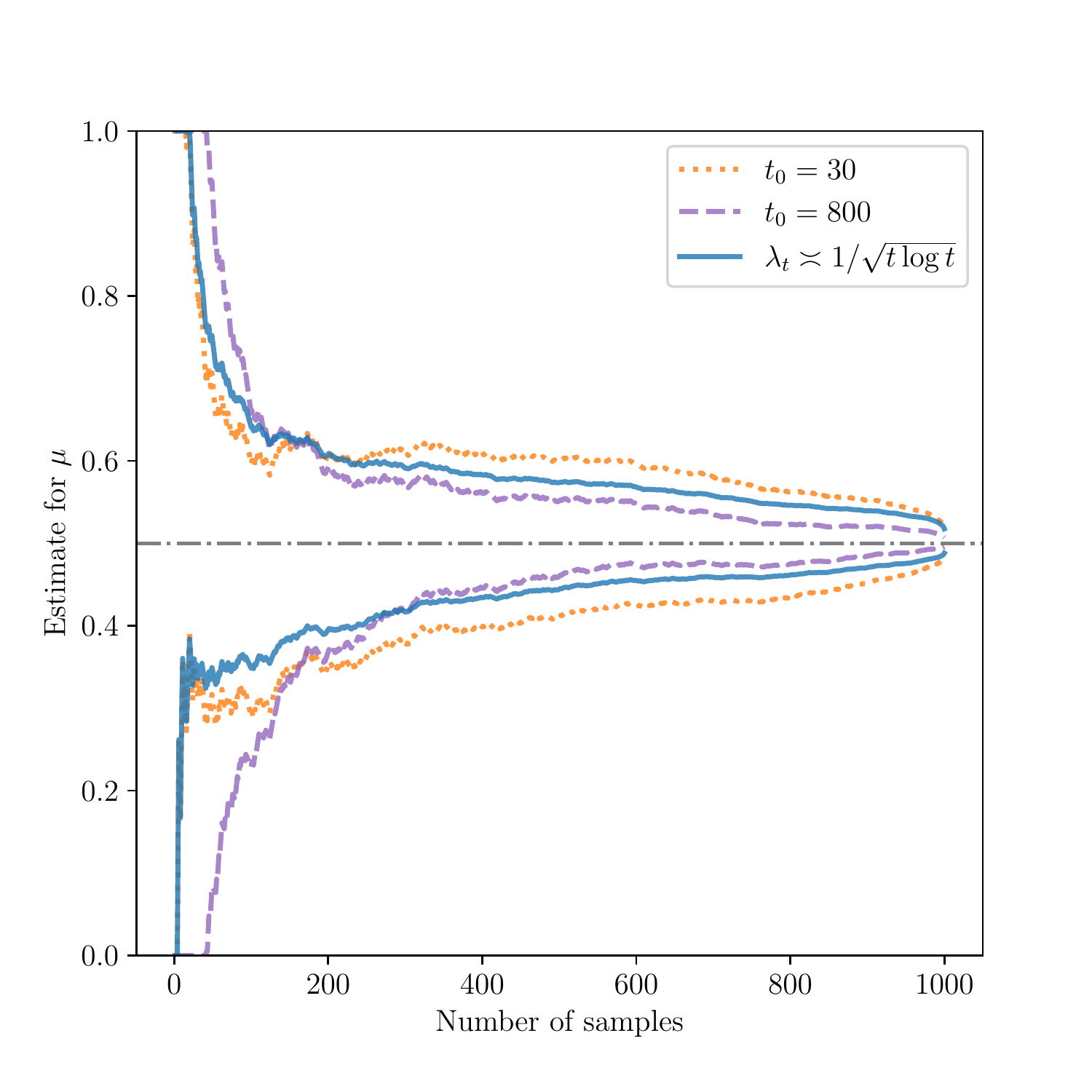}
     \caption{Hoeffding CSs based on fixed $\lambda$ values optimized for times $30$ and $800$, respectively alongside the CS based on the $\lambda$-sequence in \eqref{eq:lambdaSequenceHoeffding}. Notice that no CS uniformly dominates the others, but that the sequence in \eqref{eq:lambdaSequenceHoeffding} acts as a middle ground between the other two.}
     \label{fig:my_label}
 \end{figure}

\section{Comparing our CSs to those implied by Bardenet \& Maillard}

Bardenet \& Maillard \cite[Theorem 2.4]{bardenet2015concentration} provide the following two time-uniform Hoeffding-Serfling inequalities when sampling bounded real numbers WoR from a finite population. For any $n \in [N]$,
\begin{align*}
 \Pr \left ( \exists t \in \{ 1, \dots, n\} : \frac{1}{N-t} \sum_{i=1}^t(X_i - \mu) \geq \frac{n\epsilon}{N-n} \right ) &\leq \exp \left \{ -\frac{2n\epsilon^2 }{(1-(n-1)/N) (u-\ell)^2} \right \} \quad \text{and} \\
 \Pr \left (\exists t \in \{n, \dots, N-1\} : \frac{1}{t}\sum_{i=1}^t (X_i - \mu) \geq \epsilon \right ) &\leq \exp \left \{ -\frac{2n\epsilon^2}{(1-n/N)(1+1/n)(u-\ell)^2} \right \}.
 \end{align*}
Inverting these inequalities and taking a union bound to get two-sided inequalities, we have
\begin{align}\label{eq:BM-CS}
    \frac{1}{t}\sum_{i=1}^t X_i &\pm \frac{n(N-t)}{t(N-n)} \sqrt{\frac{\log (4/\alpha) (1-(n-1)/N) (u-\ell)^2 }{2n}} & \text{ when } t \leq n\\
    \frac{1}{t}\sum_{i=1}^t X_i &\pm \sqrt{ \frac{\log (4/\alpha)(1-n/N)(1+1/n)(u-\ell)^2 }{2n} } & \text{ when } t \geq n
    \label{eq:BM-CS2}
\end{align}
is a $(1-\alpha)$ CS for $\mu$. We term the CS defined by \eqref{eq:BM-CS} and \eqref{eq:BM-CS2} as the Bardenet-Maillard CS for simplicity.

A comparison of the aforementioned CS to our Hoeffding-type CS is displayed in Figure~\ref{fig:HSvsOurs}, where we see that our bound is roughly as tight as the Bardenet-Maillard CS at the time of optimization, while our bounds are (much) tighter everywhere else. This phenomenon was observed and studied in the with-replacement setting, attributing the benefits of confidence bounds like our Hoeffding CS to an underlying `line-crossing' inequality being uniformly tighter than an underlying Freedman-type inequality. For more information on the with-replacement analogy, we direct the reader to the pair of papers by Howard et al. \cite{howard2018uniform, howard_exponential_2018}. Returning back to the WoR setting, we remark that \eqref{eq:BM-CS} uses the standard sample mean, but we use a more sophisticated sample mean \eqref{eqn:withoutReplaceMeanEst}.

\begin{figure}[!htb]
    \centering
    \includegraphics[width=0.8\textwidth]{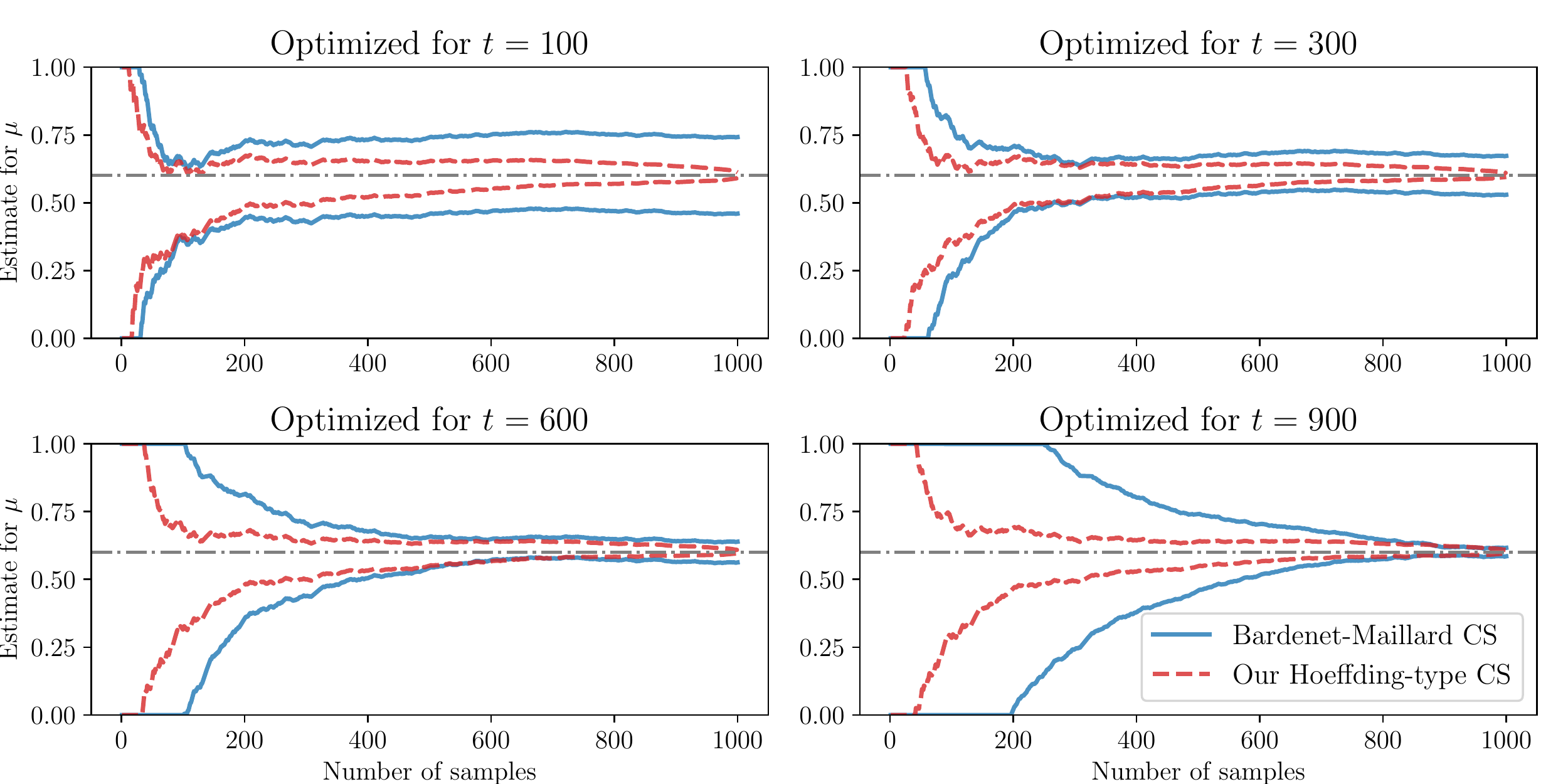}
    \caption{A comparison of our Hoeffding-type CS against the Hoeffding-Serfling CS of Bardenet \& Maillard \cite{bardenet2015concentration}. Our Hoeffding CS appears to be as tight as the Hoeffding-Serfling bound at the time of optimization, but tighter at all other times.}
    \label{fig:HSvsOurs}
\end{figure}

\newpage
\section{Time-uniform versus fixed-time bounds}

A natural question to ask is, `how much does one sacrifice by using a time-uniform CS instead of a fixed-time confidence interval'? The answer to this question will depend largely on the type of bound used, the underlying finite population, and other factors. However, in the case of sampling binary numbers from a finite population, it seems that the answer is `not much'. In Figure~\ref{fig:timeUniformFixedTime}, we display the fixed-time Hoeffding confidence interval of Corollary~\ref{corollary:hoeffdingCI} alongside its time-uniform counterpart from Theorem~\ref{theorem:hoeffdingConfseq} and the prior-posterior ratio CS from Theorem~\ref{theorem:hyperGeoConfseq}. In terms of the width of confidence bounds, we find that not much is lost by using the two aforementioned CSs over the fixed-time Hoeffding confidence interval. For this small price, the user is awarded the flexibility that comes with using CSs such as properties (a), (b), and (c) described in the Introduction.

\begin{figure}[!htb]
    \centering
    \includegraphics[width=0.75\textwidth]{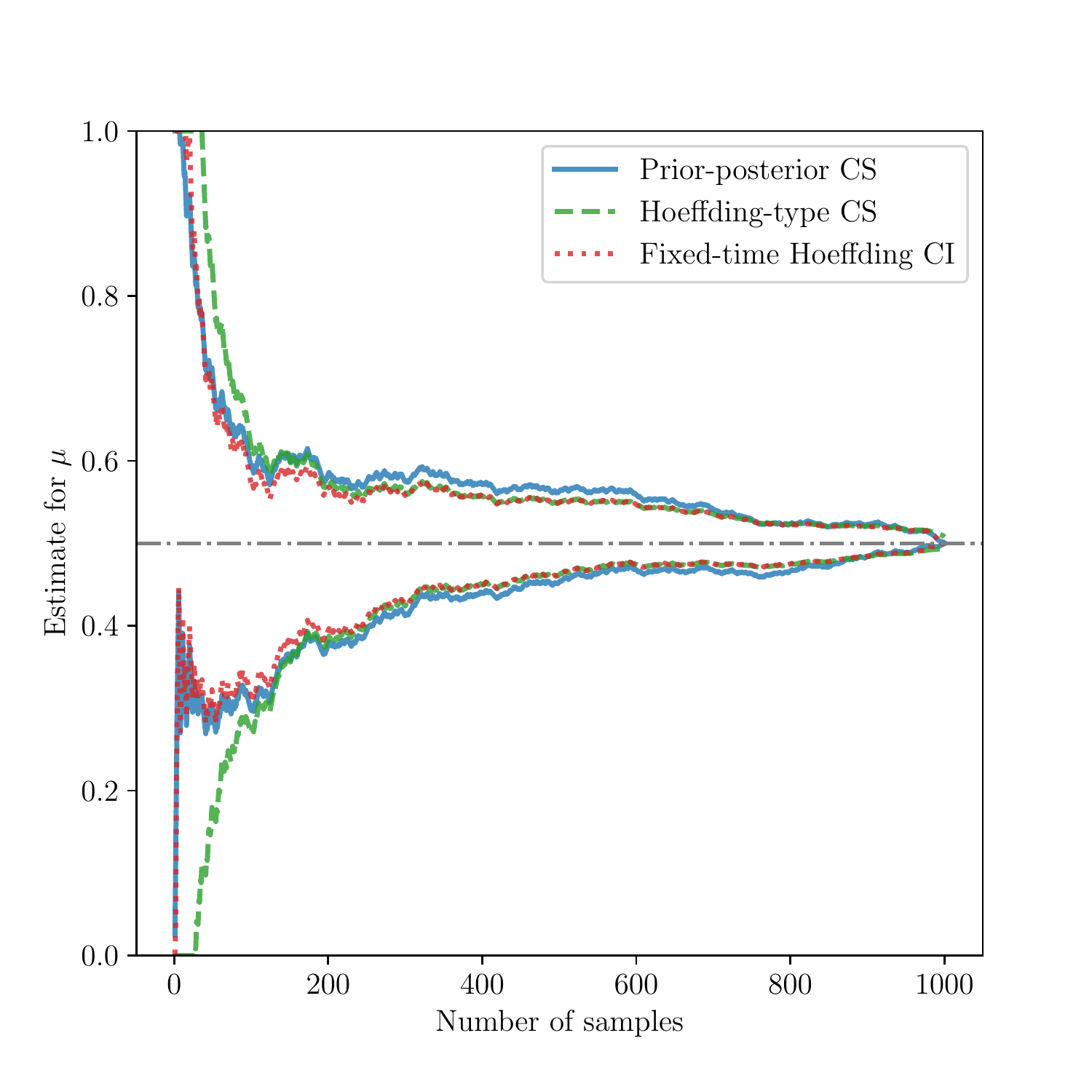}
    \caption{Comparing fixed-time and time-uniform confidence bounds for sampling binary numbers from a population of size 1000, consisting of 500 ones and 500 zeros. The dotted red line shows the fixed-time Hoeffding bound of Corollary~\ref{corollary:hoeffdingCI}, while the dashed green and solid blue lines refer to the time-uniform Hoeffding-type CS and the prior-posterior ratio CS, respectively. Notice that the increase in confidence bound width that results from using a time-uniform bound is relatively minor.}
    \label{fig:timeUniformFixedTime}
\end{figure}

\section{Computational considerations}

When using the CSs of Theorems~\ref{theorem:hyperGeoConfseq}, \ref{theorem:hoeffdingConfseq}, and \ref{theorem:empiricalBernsteinConfseq} in practice, it is important to keep in mind the computational costs associated with each method. For fixed values of $\lambda$, updating the Hoeffding and empirical Bernstein CSs at a each time $t$ takes constant time and constant memory, since all calculations involve cumulative sums (or averages). Furthermore, optimal values of $\lambda$ can be computed as in \eqref{eqn:lambdaOpt} for Hoeffding-type bounds and approximated as in \eqref{eqn:lambdaOpt2} for empirical Bernstein-type bounds, all in constant time.
On the other hand, the prior-posterior ratio (PPR) CS of Theorem~\ref{theorem:hyperGeoConfseq} is the more computationally expensive method among those presented, but can still be computed quickly for many problems. In order to find the CS, 
\[ C_t :=  \left \{ n^+ \in [N] : \frac{\pi_0(n^+)}{\pi_t(n^+)} < \frac{1}{\alpha} \right \}, \]
one must find all values in $\{ 0, \dots, N\}$ which, when provided as an input to $\frac{\pi_0(\cdot )}{\pi_t(\cdot )}$ are less than $1/\alpha$. 

Therefore, computing the entire CS takes $O(PN^2)$ time where $P$ is the time required to compute $\pi_0(n) / \pi_t(n)$. In all of the PPR CSs presented in this paper, we used computationally tractable conjugate priors, so $P=1$. We believe more sophisticated root-finding methods can be employed to arrive at a time of $O(N \log(PN))$, but these methods are reasonably fast in our experience. Moreover, the PPR CS can be computed on a subset of $[N]$ if needed, and is parallelizable. 

For reference, we provide average computation times in Table~\ref{tab:computationTimes}. All calculations were measured using Python's default \texttt{time} package and were performed in Python 3.8.3 using the \texttt{numpy} and \texttt{scipy} packages on a quad-core CPU with 8 threads at 1.8GHz each. However, no parallel processing was performed aside from the default multithreading provided by Python. 
\begin{table}[!htbp]
\centering
\begin{tabular}{|ll|}
\hline
                      & Time in seconds (std. dev.)                   \\
                      \hline
Hoeffding             & $2.13 \times 10^{-4}$ ($2.88\times 10^{-5}$)  \\
Empirical Bernstein   & $2.35 \times 10^{-4}$ ($3.24 \times 10^{-5}$) \\
Prior-posterior ratio & $0.306$ ($0.0115$) \\                 
\hline
\end{tabular}
\caption{Average time taken to compute the various CSs for $N = 1000$ discrete observations with equal numbers of ones and zeros, with standard deviations for 100 repeated experiments.}
\label{tab:computationTimes}
\end{table}

\section{Simple experiments for computing miscoverage rates}

\begin{figure}[!htb]
    \centering
    \includegraphics[width=\textwidth]{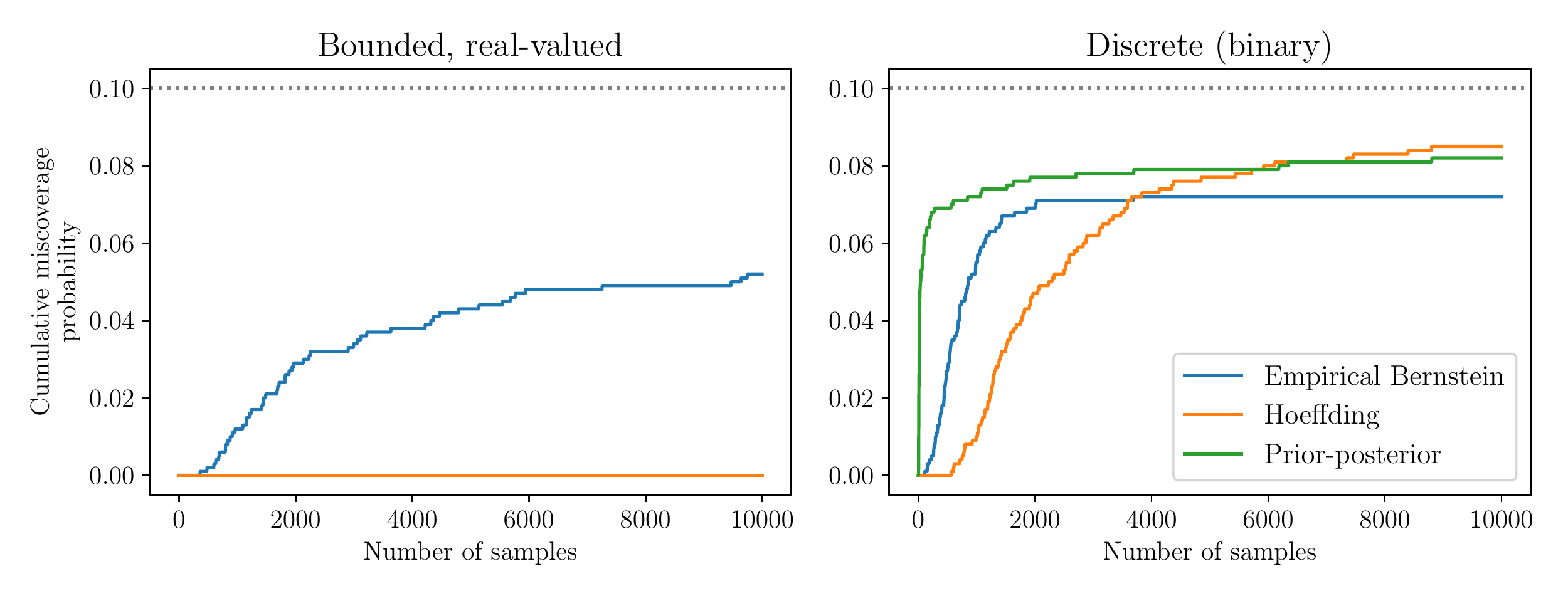}
    \caption{Empirical miscoverage probabilities for our empirical Bernstein, Hoeffding, and prior-posterior CSs. The left plot compares empirical Bernstein and Hoeffding for a population of $N=10,000$ consisting of bounded, real-valued observations uniformly distributed on the unit interval. The plot on the right-hand side compares all three for a population of the same size containing discrete elements with zeros and ones in equal proportions. Notice that while the empirical Bernstein CS does reasonably well in both settings, none of the three methods uniformly dominates the others.}
    \label{fig:cumulativeRejectionProb}
\end{figure}

Typically, in nonparametric testing, there is no `uniformly most powerful' test: any test achieving high power against some class of alternatives must necessarily be less powerful against some other class of alternatives, while a different test may display the opposite behavior. An analogous story holds for nonparametric estimation as well: the class of bounded random variables (or sequences of bounded random numbers) is nonparametric, and in such a setting, no single estimation technique can uniformly dominate all others (that is, always have lower width for any bounded sequence). This phenomenon is easy to exemplify for our confidence sequences: we can construct settings where the Hoeffding-type CS is less conservative (tighter estimation, more powerful as a test) than the empirical-Bernstein CS, and other settings in which the opposite is true.  Figure~\ref{fig:cumulativeRejectionProb} considers two such `opposite' scenarios: the binary setting which maximizes the variance of the sequence, and another setting in which the observations are uniformly distributed on $[0,1]$. In the first setting, there is no point in `estimating' the variance (empirical-Bernstein) as opposed to just assuming that it is the maximum possible variance (Hoeffding-type), and so the former is more conservative than the latter. In the second setting, the Hoeffding CS is far more conservative, as expected. With no prior knowledge on the type of sequence to be encountered, the empirical Bernstein CS seems like a safer choice.

\end{document}